\documentclass[12pt,reqno]{amsart}

\usepackage{amsmath,amssymb,amsthm,mathrsfs}
\usepackage{graphicx,cite,times}
\usepackage{cases}
\usepackage{color}
\usepackage{appendix}
\usepackage{enumerate}
\usepackage{bm} 
\usepackage{ulem}

\usepackage{mathtools}
\usepackage[ruled]{algorithm2e}
\usepackage{soul}
\usepackage{verbatim}
\usepackage{subfigure} 
\usepackage[colorlinks, linkcolor=blue, citecolor=blue]{hyperref}

\usepackage{cite}
\graphicspath{{figures/}}
\usepackage{epstopdf}
\usepackage{rotating}
\usepackage{algorithmic}

\newtheorem{theo}{Theorem}[section]
\newtheorem{coro}[theo]{Corollary}
\newtheorem{lemm}[theo]{Lemma}
\newtheorem{prop}[theo]{Proposition}
\newtheorem{rema}[theo]{Remark}
\newtheorem{defi}[theo]{Definition}

\newtheorem{assu}{Assumption}
\newtheorem{exam}[theo]{Example}

\numberwithin{equation}{section}

\begin{document}

\title[Gel'fand's Inverse Problem on connected weighted graphs]{Vertex Weight
Reconstruction in the Gel'fand's Inverse Problem on Connected Weighted Graphs}

\author[Li]{Songshuo Li}
\address{School of Mathematics and Statistics, and Center for Mathematics
and Interdisciplinary Sciences, Northeast Normal University, Changchun,
Jilin 130024, P.R.China}
\email{liss342@nenu.edu.cn}

\author[Gao]{Yixian Gao}
\address{School of Mathematics and Statistics, and Center for Mathematics
and Interdisciplinary Sciences, Northeast Normal University, Changchun,
Jilin 130024, P.R.China}
\email{gaoyx643@nenu.edu.cn}

\author[Geng]{Ru Geng}
\address{School of Mathematics and Statistics, and Center for Mathematics
and Interdisciplinary Sciences, Northeast Normal University, Changchun,
Jilin 130024, P.R.China}
\email{gengru93@163.com}

\author[Yang]{Yang Yang}
\address{Department of Computational Mathematics Science and Engineering, Michigan State University, East Lansing, MI 48824, USA}
\curraddr{}
\email{yangy5@msu.edu}

\thanks{ The research of Y. Gao was  supported by NSFC grants (project numbers, 12371187, 12071065)  and Science and Technology Development Plan Project of Jilin Province
20240101006JJ.
The research of Y. Yang is partially supported by the NSF grants DMS-2006881, DMS-2237534, DMS-2220373, and the NIH grant R03-EB033521.}


\begin{abstract}
We consider the reconstruction of the vertex weight in the discrete Gel'fand's inverse boundary spectral problem for the graph Laplacian.
Given the boundary vertex weight and the edge weight of the graph, we develop reconstruction procedures to recover the interior vertex weight from the Neumann boundary spectral data on a class of finite, connected and weighted graphs.
The procedures are divided into two stages: the first stage reconstructs the Neumann-to-Dirichlet map for the graph wave equation from the Neumann boundary spectral data, and the second stage reconstructs the interior vertex weight from the Neumann-to-Dirichlet map using the boundary control method adapted to weighted graphs. For the second stage, we identify a class of weighted graphs where the unique continuation principle holds for the graph wave equation. 
The reconstruction procedures are further turned into an algorithm, which is implemented and validated on several numerical examples with quantitative performance reported.
\end{abstract}

\maketitle

\section{Introduction and Main Results}

The Gel'fand's inverse boundary spectral problem aims to determine a differential operator based on the knowledge of its boundary spectral data \cite{MR0095423}.
This problem arises in various scientific and engineering domains where understanding the internal structure of a system or material is crucial. In this paper, we are interested in the discrete Gel'fand's inverse boundary spectral problem on combinatorial graphs~\cite{MR4620352}. 
In the discrete formulation, traditional differential operators are substituted with difference operators, and traditional functions are substituted with functions defined on vertices. The problem thus involves reconstructing properties of combinatorial graphs from  boundary spectral data. The analysis of this discrete problem serves as a foundational framework for finite difference and finite element analysis of numerical methods for solving the continuous Gel'fand's inverse boundary spectral problem.

We formulate the discrete Gel'fand's inverse boundary spectral problem following the presentation in~\cite{MR4620352}. A graph $(\bar{G},\mathcal{E})$ consists of a set of vertices $\bar{G}$ and a set of edges $\mathcal{E}$. The set of vertices admits a disjoint decomposition $\bar{G}=G\cup\partial{G}$, where $G$ is called the set of \textit{interior vertices} and $\partial G$ the set of \textit{boundary vertices}.
The graph is \textit{finite} if $|\bar{G}|$ and $|\mathcal{E}|$ are both finite, where $|\cdot|$ denotes the cardinality. 
Given two vertices $x,y\in \bar{G}$, we say that $x$ is a neighbor of $y$, denoted by $x\sim y$ , if there exists an edge connecting $x$ and $y$. This edge is denoted by $\{x, y\}$. In this case, $y$ is clearly a neighbor of $x$ as well.
The graph is \textit{undirected} if the edges do not carry directions, that is, if $\{x, y\} = \{y,x\}$.
The graph is \textit{weighted} if there exists an edge weight function on the set of edges $w: \mathcal{E}\longrightarrow \mathbb{R}_+$ ($\mathbb{R}_+$ denotes the set of positive real numbers) such that  $w(x,y)=w(y,x)>0$ for $x\sim y$.  By convention, if there is no edge between $x$ and $y$, we set $w(x,y)=w(y,x)=0$. 
We often  use the simplified notation $w_{x,y}$ to represent the edge weight $w(x,y)$ for brevity.
The graph is \textit{simple} if there is at most one edge between any two vertices and no edge connects a vertex to itself.
The graph is \textit{connected} if any two vertices can be connected by a sequence of edges.
In this paper, all the graphs are assumed to be finite, undirected, weighted, simple and connected.

For $x\in\bar{G}$,  its \textit{degree}, denoted by $\deg(x)$, is defined as
the number of edges in $\mathcal{E}$ connecting it to its neighbors. 
Let $u:\bar{G}\longrightarrow \mathbb{R}$ be a real-valued function.  
The \textit{graph Laplacian} $\Delta_G$ is defined as 
\begin{equation}
\label{defgraph Laplacian operator}
\Delta_G u(x):= \frac{1}{\mu(x)}\sum_{\substack{y\in \bar{G}\\ y\sim x}}w(x,y)(u(y)-u(x)),\quad x\in G.
\end{equation}
Here $\mu: \bar{G}\longrightarrow \mathbb{R}_+$ is a positive function on the set of vertices. 
We will refer to $\mu$ as the \textit{vertex weight} and write $\mu(x)$ as $\mu_x$ for simplicity. 
This definition of the graph Laplacian includes several special cases that are of importance in graph theory. For instance, the combinatorial Laplacian corresponds to $\mu \equiv 1$, while the normalized Laplacian corresponds to $w \equiv 1$ and $\mu(x) = {\rm deg} (x)$.
The Neumann boundary value of $u$ is defined as
\begin{align}\label{def_Neumann boundary value}
\partial_\nu u(z) := \frac{1}{\mu_z}\sum_{\substack{x\in G\\ x\sim z}}w(x,z)(u(x) - u(z)),\quad z\in \partial G.
\end{align}

A function $\varphi:\bar{G}\rightarrow\mathbb{R}$ is said to be \textit{harmonic} if 
\begin{align*}
\Delta_G \varphi (x)=0,\quad x\in G.
\end{align*}
Although the definition of $\Delta_G$ involves $\mu$,  it is clear that the concept of harmonic functions is independent of  $\mu$.
Denote by $l^2(G)$ the $l^2$-space of real-valued functions equipped with the following inner product:
for functions $u,v: G\longrightarrow \mathbb{R}$, 
\begin{equation*}
(u,v)_G := \sum\limits_{x\in G}\mu_xu(x)v(x).
\end{equation*}
Similarly, 
denote by $l^2(\partial G)$ the $l^2$-space of real-valued functions equipped with the inner product
\begin{align*}
(u,v)_{\partial G} := \sum\limits_{z\in \partial G}\mu_zu(z)v(z)
\end{align*}
for functions $u,v: \partial G\rightarrow \mathbb{R}$.

\begin{defi}[Neumann boundary spectral data]
For the Neumann eigenvalue problem
\begin{equation}\label{Neumann boundary spectral data}
\begin{aligned}
-\Delta_G\phi_j(x) & =\lambda_j\phi_j(x), \quad x\in G, \\
\partial_{\nu}{\phi_j}|_{\partial G} & =0,
\end{aligned}
\end{equation}
we say that $\phi_j$ with $(\phi_j,\phi_j)_G=1$ is a normalized Neumann eigenfunction associated to the Neumann eigenvalue $\lambda_j$.
The collection of the eigenpairs $\{(\lambda_j$,$\phi_j|_{\partial G})\}_{j=1}^{|G|}$ is  called the Neumann boundary spectral data.
\end{defi}

\begin{rema}
The graph Laplacian equipped with the homogeneous Neumann boundary condition is self-adjoint (see Lemma \ref{Green's formula on graph}), hence all the Neumann eigenvalues are real, and the normalized Neumann eigenfunctions $\{\phi_j(x) \mid x\in G\}_{j=1}^{|G|}$ form an orthonormal basis of $l^2(G)$.
\end{rema}

The discrete Gel’fand inverse spectral problem concerns reconstruction of the interior vertex set $G$, the edge set $\mathcal{E}$, and the weight functions $w, \mu$ from the Neumann boundary spectral data~\cite{MR4620352}. However, it is worth noting that solving the discrete Gel’fand’s inverse problem on general graphs is not possible due to the existence of isospectral graphs, see \cite{MR926481,MR1710481,MR1664212}. 
In this article, we restrict ourselves to the following special case: Suppose the edge weight function $w$ is known. Given the Neumann boundary spectral data and the boundary vertex weight $\mu|_{\partial G}$, what can be concluded regarding the interior vertex weight $\mu|_G$?
In \cite[Theorem 2]{MR4620352}, it is proved that $\mu$ can be uniquely determined under suitable assumptions on the graph, provided $\mu|_{\partial G}$ is known. 
However, this proof is non-constructive and does not yield explicit reconstruction. The main objective of this paper is to provide constructive procedures  for identifying $\mu|_G$, enabling the derivation of an algorithm to numerically compute $\mu|_G$.

Our constructive proof and algorithm are rooted in the boundary control method pioneered by Belishev~\cite{MR924687}, tailored for application to combinatorial graphs. An important step of the method links boundary spectral data with wave equations. Therefore, we pause here to formulate the graph wave equation following \cite{MR4620352}.
For a function $u:\mathbb{N}\times \bar{G}\longrightarrow\mathbb{R}$, we define the discrete first and second time derivatives as 
\begin{align*}
D_{t} u(t,x)=&u(t+1,x)-u(t,x),\qquad t\in\{0,1,\cdots\},x\in \bar{G},\\
D_{tt} u(t,x)=&u(t+1,x)-2u(t,x)+u(t-1,x),\qquad t\in\{1,2,\cdots\},x\in \bar{G}.
\end{align*}
We will refer to the following equation as the \textit{graph wave equation}:
\begin{align*}
D_{tt} u(t,x) - \Delta_G u(t,x) = 0, \qquad t\in\{1,2,\cdots\},  x\in G.
\end{align*}

Our first goal is to prove a unique continuation result for the graph wave equation. To this end, we introduce some terminologies.
Given any $x,y\in\bar{G}$, their \textit{distance}, denoted by $d(x,y)$, is defined as the minimum number of edges that connect $x$ and $y$ via other vertices.
For $x\in \bar{G}$ , its \textit{distance to the boundary} $\partial G$ is defined as
\begin{align*}
d(x,\partial G)= \min\limits_{z\in\partial G}d(x,z),\quad x\in \bar{G}.
\end{align*}
We say a vertex $x\in \bar{G}$ has \textit{level} $l$ if $d(x,\partial G)=l$. 
Obviously, $l$ is an integer and $0\leq l\leq \max\limits_{x\in G}d(x,\partial G)$.
The collection of interior vertices of \textit{level} $l$ is denoted by
\begin{align*}
N_l := \{x\in G\mid d(x,\partial G)=l\}.
\end{align*}
For a subset of vertices $\Omega\subset \bar{G}$, the set
$$\mathcal{N}(\Omega)= \{y\in \bar{G}~|~x\sim y, ~x\in \Omega\}$$
is called the \textit{neighborhood} of $\Omega$ in $\bar{G}$.
If there exists $y_0\in \bar{G}$ such that $y_0\in \mathcal{N}(x)\cap N_{l+1}$ for $x\in N_l$, then $y_0$ is called a \textit{next-level neighbor} of $x$.

The following assumption on the topology of $\bar{G}$ is critical for our proof of the unique continuation result.

\begin{assu}\label{foliation condition} 
\begin{enumerate}[(i)]
	\item Every boundary vertex connects to a unique interior vertex.
	\item For each integer $l$ with $1\leq l\leq \max\limits_{x\in G}d(x,\partial G)$, the set of vertices of level $l$ admits the decomposition
\begin{align*}
N_l=\bigcup\limits_{r=1}^{k_l}N_l^{r} \quad \text{ for }~k_l\in \mathbb{N_+},
\end{align*}
where  $k_l\in \mathbb{N_+}$ depends on $l$, and the sets $N^1_l$ and $N^k_l$ are defined as
\begin{align*}
N_l^{1}:&=\{x\in N_l: |\mathcal{N}(x)\cap N_{l+1}|\leq1\},\\
N_l^{k}:&=\left\{x\in N_l: |\mathcal{N}(x)\cap N_{l+1}|>1,~ |\mathcal{N}(x)\cap N_{l+1}\setminus\left(\bigcup\limits_{r=1}^{k-1}\mathcal{N}({N_l^{r}})\right)|\leq1\right\} ,\quad 2\leq k\leq k_l.
\end{align*}
\end{enumerate}
\end{assu}

\begin{rema}
Assumption\ref{foliation condition}(ii) means that every $N_l$ consists of two types of vertices: the first type are those in $N_l^1$, they have at most one next-level neighbor; the second type are those in $N_l^k$ ($k = 2, \dots, k_l$), they may have multiple next-level neighbors but at most one of them is not a neighbor of any vertices in $N_l^1,\dots,N_l^{k-1}$. This may be viewed as a type of ``foliation condition'' for the graph.
\end{rema}

The first main result of this paper is the following unique continuation property for the graph wave equation.

\begin{theo}\label{unique_continuation_principle}(Unique continuation theorem).
Suppose $\bar{G}$ satisfies Assumption~\ref{foliation condition}.
If $u(t,x)$ satisfies the graph wave equation with vanishing Dirichlet and Neumann data:
\begin{align*} 
\begin{cases}
D_{tt} u(t,x) - \Delta_G u(t,x) = 0, \quad & (t,x)\in \{-T+1,-T+2,\cdots, T-1\}\times G, \\
u(t,z)=\partial_\nu u(t,z)  = 0, \quad  &  (t,z)\in \{-T,-T+1,\cdots, T\}\times\partial G,
\end{cases}
\end{align*}
then
\begin{equation} \label{eq:uniquecont}
    u(t,x)=0,\qquad (t,x)\in\{-T+l-1,\cdots,T-l+1\}\times N_l
\end{equation}
for all $l = 1,2,\dots, \max\limits_{x\in G}d(x,\partial G)$. 
In particular, if $T\geq \max\limits_{x\in G}d(x,\partial G)$, then 
$$u(0,x)=D_tu(0,x)=0$$ for all $x\in G$.
\end{theo}

Next, we consider the following initial boundary value problem for the graph wave equation: 
\begin{align} \label{eq:ibvp1}
\begin{cases}
D_{tt} u(t,x) - \Delta_G u(t,x) = 0, & (t,x)\in \{1,2,\cdots, 2T-1\}\times G, \\
u(0,x) = 0, & x\in \bar{G}, \\
D_t u(0,x) = 0, & x\in G,\\
\partial_\nu u(t,z)  = f(t,z), & (t,z)\in \{0,1,\cdots, 2T\}\times \partial G,
\end{cases}
\end{align}
where $T>0$ is an integer and $f\in l^2(\{0,1,\cdots, 2T\}\times \partial G)$ is the Neumann data. Note that we must have $f(0,z)=0$ when $z\in\partial G$, due to compatibility with the initial conditions.
This initial boundary value problem clearly has a unique solution $u=u^f$, thus we can define the Neumann-to-Dirichlet map (ND map):
\begin{align*}
\Lambda_\mu f := u^f \big|_{\{0,1,\cdots,2T\}\times\partial G}.
\end{align*}
Here, the subscript indicates that the ND map depends on the vertex weight $\mu$.

The second main result of this paper is an explicit formula to reconstruct the ND map from the Neumann boundary spectral data.

\begin{theo}\label{Neubspecd_to_NDmap}
Suppose $\bar{G}$ satisfies Assumption \ref{foliation condition} $(i)$.
Given the edge weight function $w$ and the boundary vertex weight $\mu|_{\partial G}$, the Neumann-to-Dirichlet map $\Lambda_\mu$ can be computed from the Neumann boundary spectral data $\{(\lambda_j,\phi_j\big|_{\partial G})\}^{|G|}_{j=1}$ as follows:
\begin{align}\label{formula_of_Neubsd_to_express_NDmap}
\Lambda_\mu f(t,z) = u^f(t,z) = \sum^{|G|}_{j=1} \sum\limits_{k=1}^tc_{k}(f(t+1-k),\phi_j)_{\partial G}\phi_j(z)
-\frac{\mu_zf(t,z)}{w(x,z)},\quad x\sim z 
\end{align}
when $t = 1,2,3,\dots$
Here $x\in G$ is the unique interior vertex that is connected to $z\in\partial G$, and the coefficients $c_k$ satisfies $c_1=0$, $c_2=-1$, and the recursive relation $c_k=(2-\lambda_k)c_{k-1}-c_{k-2}$ for $k\geq3$.
\end{theo}

Next, denote by $M$ the vector space spanned by products of harmonic functions, that is,
$$
M :=  {\rm span}\{\varphi\psi|_G: \Delta_G \varphi = \Delta_G \psi  = 0 \text{ in } G\}.
$$
The third main result of the paper gives an explicit reconstruction formula to obtain the orthogonal projection of $\mu|_{G}$ onto $M$ from the ND map.

\begin{theo}\label{orthogonal_projection_of_mu_onto_M}
Suppose $\bar{G}$ satisfies Assumption \ref{foliation condition}. Given the edge weight function $w$, the boundary vertex weight $\mu|_{\partial G}$, and $T\geq \max_{x\in G}d(x,\partial G)$, then the orthogonal projection of $\mu|_{G}$ onto $M$ can be explicitly reconstructed from the Neuman-to-Dirichlet map $\Lambda_\mu$. Moreover, a reconstruction algorithm is derived in Algorithm~\ref{alg:Framwork} in Section \ref{sec:Reconstruction Algorithm}.
\end{theo}

Note that Theorem~\ref{orthogonal_projection_of_mu_onto_M} only ensures reconstruction of an orthogonal projection of $\mu|_G$. In order to obtain the full interior vertex measur $\mu|_G$, further conditions have to be imposed on $G$ and the edge weight function $w$. Note that an edge weight function $w: \mathcal{E}\longrightarrow \mathbb{R}_+$ can be identified with a point in the space $\mathbb{R}_+^{|\mathcal{E}|}$ by indexing the edges in $\mathcal{E}$.

\begin{coro} \label{thm:density}
Suppose $\bar{G}$ satisfies $\frac{|\partial G|(|\partial G|+1)}{2}\geq |G|$, and suppose there exists at least one edge weight function $w$ such that $M=l^2(G)$, then $M=l^2(G)$ holds for all edge weight functions $w$ except for a set of measure zero in $\mathbb{R}_+^{|\mathcal{E}|}$.
Therefore, under the assumptions of Theorem~\ref{orthogonal_projection_of_mu_onto_M} and Corollary~\ref{thm:density}, $\mu|_{G}$ can be explicitly reconstructed from the Neuman-to-Dirichlet map $\Lambda_\mu$ for all edge weight functions except for a set of measure zero in $\mathbb{R}_+^{|\mathcal{E}|}$. In this case, Algorithm~\ref{alg:Framwork} in Section \ref{sec:Reconstruction Algorithm} recovers $\mu|_G$. 
\end{coro}

Combining Theorem~\ref{Neubspecd_to_NDmap}, Theorem~\ref{orthogonal_projection_of_mu_onto_M} and Corollary~\ref{thm:density}, we see that the vertex weight $\mu|_G$ can be constructed from the Neumann boundary spectral data for a class of graphs.

The Gel'fand's inverse boundary spectral problem for partial differential operators in the continuum setting has been extensively investigated, e.g, in \cite{MR1390624,MR1608496,MR2053425,MR1484000,MR2221698,MR2096795,Burago2020StabilityOT,MR933457,MR1127097,MR3019446}. In particular, Belishev pioneered the boundary control method \cite{MR924687} which combined with the Tataru's unique continuation result \cite{MR1326909} determines the differential operators in $\mathbb{R}^n$. The method was further extended by Belishev and Kurylev on manifolds to determine the isometry type of a Riemannian manifold from the boundary spectral data \cite{MR1177292}.
The boundary control method for partial differential operators has since been greatly generalized (e.g, see the survey \cite{MR2353313}) and numerically implemented (e.g \cite{MR1694848, MR3483640, de2018recovery, korpela2018discrete, oksanen2022linearized, MR4369044, oksanen2024linearized}).
The Gel'fand's inverse boundary spectral problem is closely connected to several other celebrated inverse problems for wave, heat and Schr\"odinger equations \cite{MR2065431}. We refer readers to the monograph \cite{MR1889089} for a comprehensive introduction to the Gel'fand's inverse boundary spectral problem as well as its connections to other inverse problems.

The discrete Gel'fand's inverse boundary spectral problem on combinatorial graphs is formulated in \cite{MR4620352}. Assuming the ``two-points condition'' (see \cite{MR4620352} or Appendix \ref{appendixA} for the precise definition), the authors of \cite{MR4620352} proved that any two finite, strongly connected, weighted graphs that are spectrally isomorphic with a boundary isomorphism must be isomorphic as graphs. This establishes the uniqueness result for determining the graph structure (including the vertices, edges and weights) from the spectral data. 
However, the proof in \cite{MR4620352} is non-constructive and it remains unclear how to explicitly compute the graph structure. A major contribution of the current paper is the development of an algorithm based on the boundary control method that reconstructs certain quantities on a class of combinatorial graphs. We remark that the idea of the boundary control method has been adapted to solve inverse problems on certain special graphs in the earlier literature, e.g, in recovering the structures of planar trees \cite{MR2067494,MR2218385} as well as in detecting cycles in graphs \cite{MR2545980}.

Inverse spectral problems on graphs arise naturally in quantum physics. A class of graphs where these problems are well-suited is quantum graphs. A quantum graph is a metric graph that carries differential operators on the edges with appropriate conditions on the vertices. Inverse spectral problems on quantum graphs usually aim at determining graph structures or differential operators from spectral data, see e.g, \cite{MR1862642,MR2146822,MR4109187,MR2148632,MR4284864,MR2587077,MR4043329,MR2067494,MR2218385,MR4542446}.
Many other inverse problems that are closely related to inverse spectral problems have also found the counterparts on graphs. Examples include inverse conductivity problems (e.g, \cite{Curtis2000InversePF}), 
inverse scattering problems (e.g, \cite{MR2913620}), and inverse interior spectral problems (e.g, \cite{MR4598377}).

This paper's major contributions include:
\begin{itemize}
    \item A reconstruction formula and an algorithm to compute the vertex weight $\mu$. The uniqueness of the vertex weight for a class of combinatorial graphs was previously addressed in~\cite{MR4620352}, but the provided proof is non-constructive and lacks explicit computational procedures. This paper focuses specifically on reconstructing the vertex weight $\mu$. We derive an explicit reconstruction formula by converting the Neumann boundary spectral data to the Neumann-to-Dirichlet map for the graph wave equation and then adapting Belishev's boundary control method to recover $\mu$. An algorithm is subsequently derived from this formula and validated through multiple numerical experiments.
    \item New uniqueness result. A critical hypothesis for the uniqueness proof in \cite[Theorem 2]{MR4620352} is the so-called ``two-points condition" (see Appendix \ref{appendixA}), which imposes specific geometric restrictions on graphs. Consequently, the uniqueness result in~\cite{MR4620352} applies only to graphs that meet the two-points condition. This paper considers a different class of graphs, based on Assumption \ref{foliation condition}, which to some extent can be viewed as a discrete ``foliation condition". In Appendix \ref{appendixA}, we provide examples demonstrating that Assumption \ref{foliation condition} is not a special case of the two-points condition, and vice versa. This distinction ensures that the class of graphs considered in this paper is not a subclass of those in ~\cite{MR4620352}. Consequently, our reconstruction formula also implies uniqueness for a new class of graphs that satisfies Assumption \ref{foliation condition} but not the two-points condition.

    \item Unique continuation for the graph wave equation. The unique continuation principle is a crucial property of wave phenomena. For the continuum wave equation with time-independent coefficients, this principle is established in Tataru's celebrated work in \cite{MR1326909}. In this paper, we identify a class of graphs (see Assumption~\ref{foliation condition}) and prove a discrete unique continuation principle for the graph wave equation (see Theorem~\ref{unique_continuation_principle}). This result plays a central role in adapting the boundary control method to combinatorial graphs.
\end{itemize}

The paper is organized as follows:
In Section \ref{Prove the unique continuation theorem}, we prove the unique continuation principle Theorem~\ref{unique_continuation_principle} and provide several concrete examples of planar graphs that satisfy Assumption~\ref{foliation condition}.
Section \ref{Calculate the ND map from the Neumann boundary spectral data} is devoted to the proof of Theorem \ref{Neubspecd_to_NDmap}.
In Section \ref{Prove the reconstruction procedure}, we develop the discrete boundary control method and describe how to construct the orthogonal projection of the vertex weight on $M$, proving Theorem~\ref{orthogonal_projection_of_mu_onto_M}.
Section \ref{sec:uniqueandrecon} identifies a class of weighted graphs where the vertex weight can be uniquely constructed for a generic set of edge weight functions, proving Corollary~\ref{thm:density}.
The reconstruction procedures are summarized and formulated as a numerical algorithm in Section \ref{sec:Reconstruction Algorithm}. 
Finally, the resulting algorithm is validated on numerical examples with the quantitative performance reported in Section \ref{sec:Numerical Expreiments}.

\section{Proof of Theorem~\ref{unique_continuation_principle}}\label{Prove the unique continuation theorem}

This section is devoted to the proof of the unique continuation principle in Theorem \ref{unique_continuation_principle}. We also provide several graphs that satisfy Assumption~\ref{foliation condition}. These graphs are subgraphs of 2D regular tilings.

\begin{proof}
We prove the claim~\eqref{eq:uniquecont} by induction on $l = 1, 2, \cdots, \max\limits_{x\in G}d(x,\partial G)$.

\textbf{Base Case:}
For the base step $l=1$, take any $x\in N_1$.  There exists a boundary vertex $z$ such that $x\sim z$. Moreover, by Assumption \ref{foliation condition} (i), $x$ is the unique interior vertex connected to $z$.
Applying the  Dirichlet condition  $u|_{\{-T,\dots,T\}\times\partial G} =0$
and Neumann condition $\partial_\nu u|_{\{-T,\dots,T\}\times\partial G} = 0$ yields, at $z\in \partial G$,   that
\begin{equation*}
    0 = \partial_\nu u(t,z) = \frac{1}{\mu_z} w(x,z) u(t,x) \quad\text{ for }~t\in\{-T,\dots,T\}.
\end{equation*}
Hence, $u(t,x)=0$ since $\mu_z>0$ and $w(x,z)>0$. This proves the base case.

{\bf{Induction Step:}}  
For the induction step, let $l_1$ be a positive integer with $l_1\leq\max_{x\in G}d(x,\partial G)-1$. Suppose for all $l\leq l_1$, we have the inductive hypothesis
\begin{equation} \label{eq:indhypo}
u(t,x)=0, \quad (t,x)\in\{-T+l-1,\cdots,T-l+1\}\times N_l.
\end{equation}
It remains to prove the case $l=l_1+1$, that is,
\begin{equation*}
u(t,x)=0, \quad (t,x)\in\{-T+l_1,\cdots,T-l_1\}\times N_{l_1+1}.   
\end{equation*}
To this end, fix $t\in\{-T+l_1,\cdots,T-l_1\}$ and consider an arbitrary $y\in N_{l_1}$. We have $u(t,y)=0$ and $D_{tt} u(t,y)= u(t+1,y) - 2 u(t,y) + u(t-1,y) = 0$ due to the inductive hypothesis~\eqref{eq:indhypo}. The wave equation at $(t,y)$ becomes
\begin{align*}
0 = D_{tt} u(t,y) - \Delta_G u(t,y) 
= -
\frac{1}{\mu_y}\sum_{\substack{x\in N_{l_1-1}\cup N_{l_1}\cup N_{l_1+1}\\ x\sim y}}w(x,y) u(t,x).
\end{align*}
In the summation, we have $u(t,x)=0$ for $x\in N_{l_1-1} \cup N_{l_1}$ because of  the inductive hypothesis~\eqref{eq:indhypo}. 
Hence,
\begin{equation} \label{eq:sumcond}
\sum_{\substack{x\in N_{{l_1}+1}\\ x\sim y}}w(x,y)u(t,x)=0.
\end{equation}
Using this identity, we will consider the decomposition $y \in N_{l_1} = N^1_{l_1} \cup N^2_{l_1} \cup \dots \cup N^{k_{l_1}}_{l_1}$,  as stated in  Assumption \ref{foliation condition} and sequentially prove $u(t,x)=0$ for all $x\in N_{l_1+1}$.

If $y\in N^1_{l_1}$, there exists  at most one $x\in N_{l_1+1}$ connected to $y$. If no such $x$ exists, there is nothing to prove. If such an $x$ exists, the condition~\eqref{eq:sumcond} reduces to $w(x,y)u(t,x)=0$, hence $u(t,x)=0$ since $w(x,y)>0$.
In other words, we have proved that $u(t,x)=0$ for all $x\in N_{l_1+1} \cap \mathcal{N}(N^1_{l_1})$.

If $y\in N^2_{l_1}$, then $y$ may have multiple next-level neighbors $x_1,\dots,x_{L_2} \in N_{l_1+1}$. However,   at most one of them, say $x_1$, is not in $\mathcal{N}(N^1_{l_1})$. 
In the previous paragraph, we have already proved $u(t,x_2)= \dots = u(t,x_{L_2}) = 0$. Therefore,  the condition~\eqref{eq:sumcond}  can be reduced to $w(y,x_1)u(t,x_1)=0$, which yields   $u(t,x_1)=0$. 
In other words, we  have proved that $u(t,x)=0$ for all $x\in N_{l_1+1} \cap \mathcal{N}(N^2_{l_1})$.

In general, if $y\in N^k_{l_1}$ ($k=2,\cdots,k_{l_1}$), then $y$ may have multiple next-level neighbors $x_1,\cdots,x_{L_k}$
$\in$ $N_{l_1+1}$ but at most one of them, say $x_1$, is not in $\bigcup\limits_{r=1}^{k-1} \mathcal{N}(N^r_{l_1})$. At this point, we have already proved $u(t,x_2)= \cdots = u(t,x_{L_k}) = 0$.  Hence, condition \eqref{eq:sumcond} reduces to $w(y,x_1)u(t,x_1)=0$ and consequently $u(t,x_1)=0$.
In other words, we have proved that $u(t,x)=0$ for all $x\in N_{l_1+1} \cap \mathcal{N}(N^k_{l_1})$.

This completes the proof that $u(t,x)=0$ for all $x\in N_{l_1+1}$, because any such $x$ must be connected to a vertex $y\in N^{k}_{l_1}$ for some $k$, that is, $x\in N_{l_1+1} \cap \mathcal{N}(N^k_{l_1})$ for some $k$. This argument holds for any $t\in\{-T+l_1,\cdots,T-l_1\}$, hence the induction step is proved.

Finally, if $T\geq \max\limits_{x\in G}d(x,\partial G)$, then $T-l+1\geq 1$, hence $\{-1,0,1\} \subset \{-T+l-1,\cdots,T-l+1\}$. For any $x\in G$, let $l$ be its level ($1\leq l \leq \max\limits_{x\in G}d(x,\partial G)$), then
\begin{equation*}
(0,x), (1,x) \in \{-T+l-1,\cdots,T-l+1\}\times N_l.
\end{equation*}
By \eqref{eq:uniquecont}, we have $u(0,x)= u(1,x)=0$ and $D_t(0,x) = u(1,x) - u(0,x) = 0$.

\end{proof}

In the rest of this section, we provide some examples that satisfy Assumption~\ref{foliation condition}
and the condition $\frac{|\partial G|(|\partial G|+1)}{2}\geq |G|$. The latter condition is motivated by the discussion in Remark \ref{rema:verticenum}.
These examples  are special subgraphs obtained from regular tilings in $\mathbb{R}^2$.

Let $m,n$ be finite integers. We make the identification $\mathbb{R}^2\simeq \mathbb{C}$ so that the coordinates of vertices can be represented using complex numbers.
In each example, we obtain a domain $\mathcal{D}_{m,n}$ by translating a fundamental domain $D_0$ along two linearly independent directions $\vec{v}_1, \vec{v}_2$, respectively.
The vertices in $D_0$ are  translated to obtain the set of interior vertices $G$.

\begin{exam}
\textbf{The graph $R_{m,n}$ with $m,n\geq 2$.}

Take $\vec{v}_1 =1+0{\rm i}$ and $\vec{v}_2 =0+{\rm i} $.
Let $D_0\subset\mathbb{R}^2$ be the rectangular domain with the set of $4$ vertices $G_0:= \{1+{\rm i}, 2+{\rm i}, 1+2{\rm i}, 2+2{\rm i}\}$.
Define
\begin{align*}
\mathcal{D}_{m,n}:=\bigcup_{\substack{0\leq j\leq m-2\\ 0\leq k\leq n-2}}(D_0+j\vec{v}_1+k\vec{v}_2),
\end{align*}
with $j,k\in \mathbb{N}$.
The set of interior vertices is
\begin{align*}
G:= \bigcup_{\substack{0\leq j\leq m-2\\ 0\leq k\leq n-2}}
(G_0+j\vec{v}_1+k\vec{v}_2),
\end{align*}
where the corresponding  set of boundary vertices is
\begin{equation*}
\partial G := (\partial G)_L \cup (\partial G)_R \cup (\partial G)_B \cup (\partial G)_T,
\end{equation*}
with
\begin{align*}
(\partial G)_L:=& \{k\vec{v}_2\mid 1\leq k \leq n\},\qquad (\partial G)_R:=\{(m+1)\vec{v}_1+k\vec{v}_2\mid 1\leq k \leq n\},\\
(\partial G)_B:=&\{j\vec{v}_1 \mid 1\leq j \leq m\},\qquad
(\partial G)_T:=\{j\vec{v}_1 +(n+1)\vec{v}_2\mid 1\leq j \leq m\}.
\end{align*}

Note that the corner vertices
$0+0 {\rm i}$, $(m+1)+0{\rm i}$, $0+(n+1){\rm i}$,
$(m+1)+(n+1){\rm i}$ are not included in $\partial G$.
The edge set $\mathcal{E}$ is defined by assigning an edge to any pair of vertices in $\bar{G}$ that is of Euclidean distance $1$, where any two boundary vertices are not connected.
This graph is denoted by $R_{m,n}$, where $m,n$ indicate the number of interior vertices along the directions
$\vec{v}_1, \vec{v}_2$, respectively. As an example, $R_{4,3}$ is illustrated in Fig. \ref{fig_rectangular}.

\begin{figure}[htbp]
\centering
\includegraphics[scale=0.68]{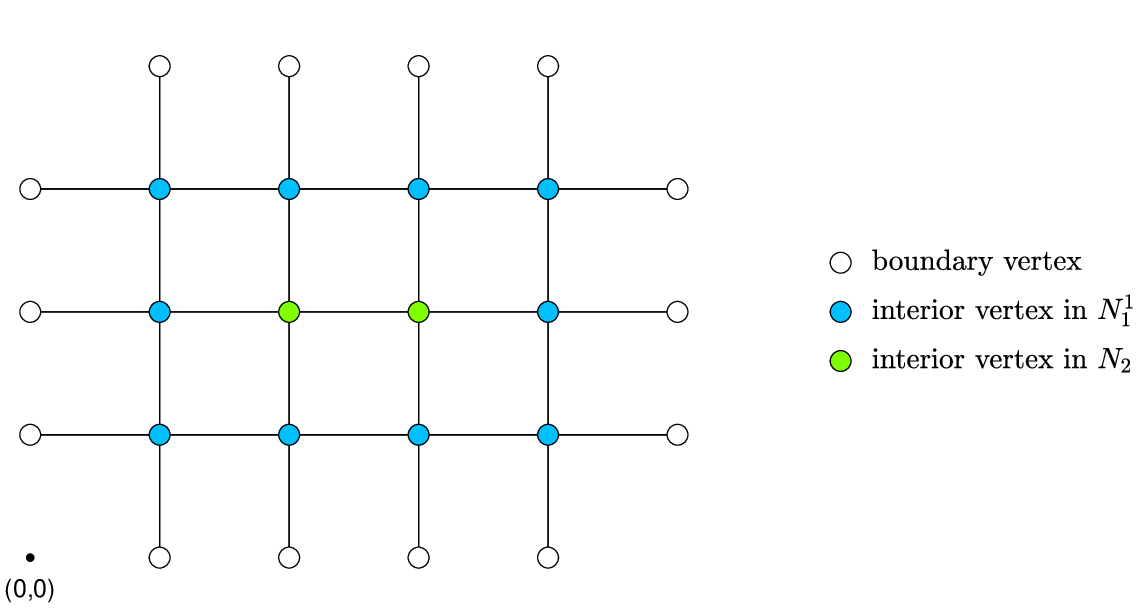}
\caption{\small\small{The graph $R_{4,3}$. }}
\label{fig_rectangular}
\end{figure}
\end{exam}

\begin{lemm}
    For any integers $m,n\geq 2$, the graph $R_{m,n}$ satisfies Assumption~\ref{foliation condition} and $\frac{|\partial G|(|\partial G|+1)}{2}\geq |G|$.
\end{lemm}
\begin{proof}
For $1\leq l\leq \max\limits_{x\in G}d(x,\partial G)$,
the set $N_l$ in $R_{m,n}$ is
\begin{align*}
N_l=&\left\{l\vec{v}_1+(l+k)\vec{v}_2 \mid 0\leq k \leq n-2l+1  \right\}
\\
&\cup \left\{(m+1-l)\vec{v}_1+(l+k)\vec{v}_2\mid 0\leq k \leq n-2l+1  \right\}\\
&\cup \left \{(l+j)\vec{v}_1+l\vec{v}_2 \mid 1\leq j \leq m-2l \right \}
\\
&\cup  \left\{(l+j)\vec{v}_1
+(n+1-l)\vec{v}_2 \mid 1\leq j \leq m-2l \right \}.
\end{align*}
The decomposition of $N_l$ is trivial as $N_l=N_l^1$.
For each vertex respectively in the above four subsets of $N_l$, there exists a boundary vertex closest to it in
$(\partial G)_L, (\partial G)_R, (\partial G)_B$ and $(\partial G)_T$ respectively.

To show the relation between the boundary and interior vertices, simply notice that $|\partial G|=2(m+n)$ and $|G|=mn$, thus $\frac{|\partial G|(|\partial G|+1)}{2} = (m+n)(2m+2n+1) \geq mn = |G|$.
\end{proof}

\begin{exam}
\textbf{The graph $T_{m,n}$ with $m,n\geq 2$.}

Take $\vec{v}_1=1+ 0{\rm i}$, $\vec{v}_2=\frac{1}{2}+ \frac{\sqrt3}{2}{\rm i}$.
Let $D_0\subset\mathbb{R}^2$ be the triangular domain with the set of 3 vertices $G_0:= \{\vec{v}_1+\vec{v}_2, 2\vec{v}_1+\vec{v}_2, \vec{v}_1+2\vec{v}_2\}$.
Define
\begin{align*}
\mathcal{D}_{m,n}:=\bigcup_{\substack{0\leq j\leq m-2\\ 0\leq k\leq n-2}}(D_0+j\vec{v}_1+k\vec{v}_2),
\end{align*}
where $j,k\in \mathbb{N}$.
The set of interior vertices is
\begin{align*}
G:= \bigcup_{\substack{0\leq j\leq m-2\\ 0\leq k\leq n-2}}
(G_0+j\vec{v}_1+k\vec{v}_2),
\end{align*}
The set of boundary vertices is $\partial G := (\partial G)_L \cup (\partial G)_R \cup (\partial G)_B \cup (\partial G)_T$ with
\begin{align*}
(\partial G)_L:=& \{k\vec{v}_2\mid 1\leq k \leq n\},\qquad (\partial G)_R:=\{(m+1)\vec{v}_1+k\vec{v}_2\mid 1\leq k \leq n\},\\
(\partial G)_B:=&\{j\vec{v}_1 \mid 1\leq j \leq m\},\qquad
(\partial G)_T:=\{j\vec{v}_1 +(n+1)\vec{v}_2\mid 1\leq j \leq m\}.
\end{align*}
The definition of the edge set $\mathcal{E}$  is as follows: an edge is assigned to any pair of vertices in $\bar{G}$ of Euclidean distance $1$, where every boundary vertex in $\partial G$ connects to an interior vertex in $G$ whose coordinates differ by vectors $\vec{v}_1$ or $\vec{v}_2$.
This graph is denoted by $T_{m,n}$, where $m,n$ indicate the number of interior vertices along the directions
$\vec{v}_1, \vec{v}_2$, respectively. As an example, $T_{6,4}$ is illustrated in Fig. \ref{fig_triangular}.

\begin{figure}[htbp]
\centering
\includegraphics[scale=0.6]{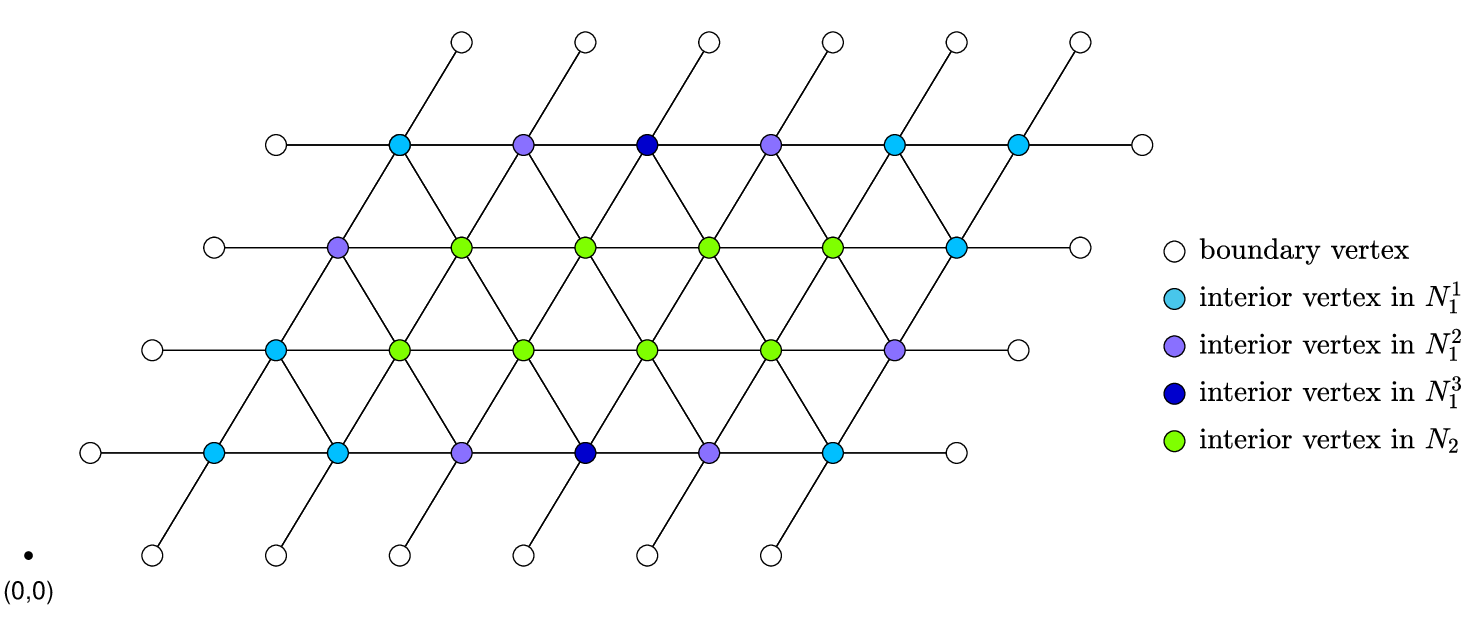}
\caption{\small\small{The graph $T_{6,4}$. }}
\label{fig_triangular}
\end{figure}
\end{exam}

\begin{lemm}
    For any integers $m,n\geq 2$, the graph $T_{m,n}$ satisfies Assumption~\ref{foliation condition} and $\frac{|\partial G|(|\partial G|+1)}{2}\geq |G|$.
\end{lemm}
\begin{proof}
For $1\leq l\leq \max\limits_{x\in G}d(x,\partial G)$,
the set $N_l$ in $T_{m,n}$ is
\begin{align*}
N_l=&\{l\vec{v}_1+(l+k)\vec{v}_2 \mid 0\leq k \leq n-2l+1 \}:=\{x_0,x_1,\cdots,x_{n-2l+1}\}\\
&\cup\{(m-l+1)\vec{v}_1
+(l+k)\vec{v}_2\mid 0\leq k \leq n-2l+1 \}:=\{y_0,y_1,\cdots,y_{n-2l+1}\}\\
&\cup \{(l+j)\vec{v}_1+l\vec{v}_2 \mid 1\leq j \leq m-2l \}:=\{\gamma_1,\cdots,\gamma_{m-2l}\}\\
&\cup \{(l+j)\vec{v}_1+(n-l+1)\vec{v}_2\mid 1\leq j \leq m-2l \}:=\{\tau_1,\cdots,\tau_{m-2l}\}.
\end{align*}
For each vertex respectively in the above four subsets of $N_l$, there exists a boundary vertex closest to it in
$(\partial G)_L, (\partial G)_R, (\partial G)_B$ and $(\partial G)_T$ respectively.
The above $x_r, y_r,\gamma_r,\tau_r$ when $r\in\mathbb{N}$ are numbers of the vertices in the sets.

Let integer $p\geq 2$.
The decomposition of $N_l$ is
\begin{align*}
N_l^1=& \{x_0,x_1,x_{n-2l+1},y_0,y_{n-2l},y_{n-2l+1},\gamma_1,\tau_{m-2l}\},\\
N_l^p=& \{x_p,x_{n-2l+2-p},y_{p-1},y_{n-2l-p+1} \mid p \leq n-2l+2-p \}\\
&\cup\{\gamma_p,\gamma_{m-2l+2-p}, \tau_{p-1}, \tau_{m-2l-p+1} \mid p\leq m-2l+2-p\},\quad p\geq 2.
\end{align*}

To show the relation between the boundary and interior vertices, simply notice that $|\partial G|=2(m+n)$ and $|G|=mn$, thus $\frac{|\partial G|(|\partial G|+1)}{2} = (m+n)(2m+2n+1) \geq mn = |G|$.
\end{proof}

\begin{exam}
\textbf{The graph $H_{m,n}$, where $m,n\geq 1$ and $m$ is odd.} 

Let $\omega=\frac{1}{2}+\frac{\sqrt3}{2}{\rm i}$, then $\omega^6=1$, where $6$ is the power of $\omega$.
Take $\vec{v}_1=3\omega^6=3+0i$, $\vec{v}_2=0+\sqrt{3}{\rm i}$.
Let $D_0\subset\mathbb{R}^2$ be the hexagon domain with the set of 6 vertices $G_0:= \{\omega^0, \omega^1, \omega^2, \omega^3, \omega^4, \omega^5\}$.
Define
\begin{align*}
\mathcal{D}_{m,n}:=\bigcup_{\substack{0\leq j\leq \frac{m-1}{2}\\ 0\leq k\leq n-1}}(D_0+j\vec{v}_1+k\vec{v}_2),
\end{align*}
where $j,k\in \mathbb{N}$.
The set of interior vertices is
\begin{align*}
G:= \bigcup_{\substack{0\leq j\leq \frac{m-1}{2}\\ 0\leq k\leq n-1}}
(G_0+j\vec{v}_1+k\vec{v}_2),
\end{align*}
The set of boundary vertices is $\partial G := (\partial G)_L \cup (\partial G)_R \cup (\partial G)_B \cup (\partial G)_T$,  where
\begin{align*}
(\partial G)_L:=& \left\{k\vec{v}_2+2\omega^3\mid 0\leq k\leq n-1  \right\},\quad
(\partial G)_R:= \left\{k\vec{v}_2+\frac{3(m-1)}{2}+2\mid 0\leq k\leq n-1  \right\},\\
(\partial G)_B:=& \left\{j\vec{v}_1+2\omega^4, j\vec{v}_1+2\omega^4+2 \mid 0\leq j\leq \frac{m-1}{2} \right\},\\
(\partial G)_T:=& \left\{
j\vec{v}_1+(n-1)\vec{v}_2+2\omega^2, j\vec{v}_1+(n-1)\vec{v}_2+2\omega^2+2 \mid 0\leq j\leq \frac{m-1}{2}  \right\}.
\end{align*}
The edge set $\mathcal{E}$ is defined by assigning an edge to any pair of vertices in $\bar{G}$ that is of Euclidean distance $1$, where any two boundary vertices are not connected.
This graph is denoted by $H_{m,n}$, where $m,n$ indicate the number of hexagons on the border along the directions
$\vec{v}_1, \vec{v}_2$, respectively.
As an example, $H_{3,4}$ and its vertices decomposition are shown in Fig. \ref{figure_hexagonal_para_addline}.
\end{exam}

\begin{figure}[htbp]
\centering
\includegraphics[scale=0.7]{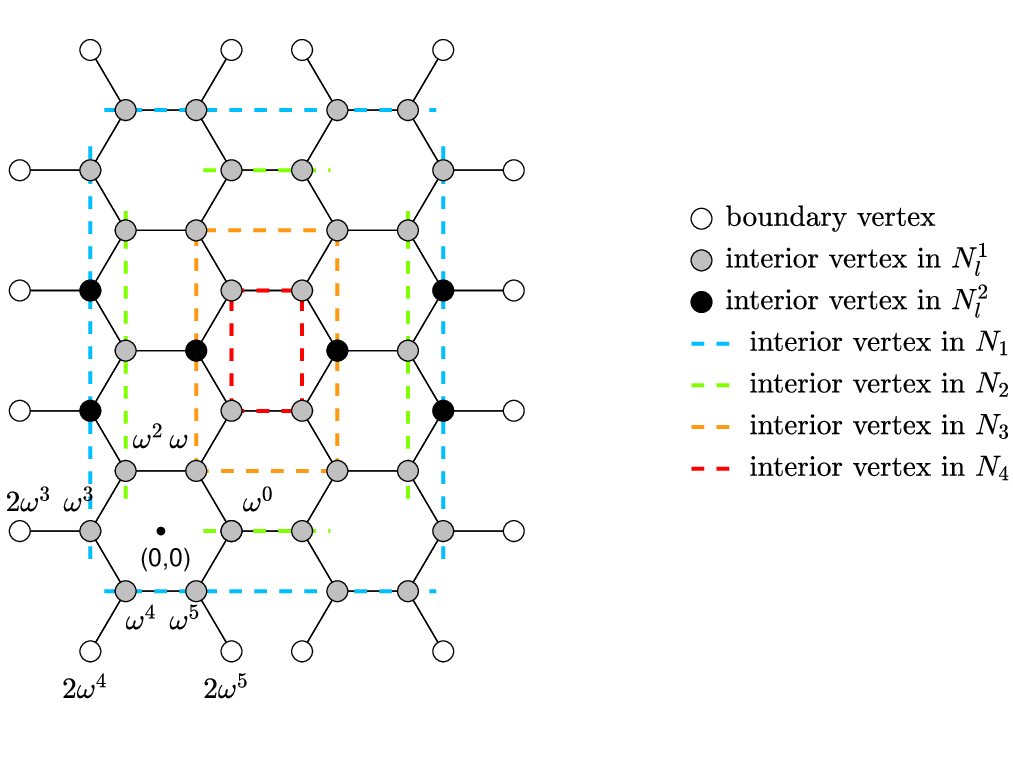}
\caption{\small\small{The graph $H_{3,4}$. }}
\label{figure_hexagonal_para_addline}
\end{figure}

\begin{lemm}
For any integers $m,n\geq 1$ and $m$ is odd, the graph $H_{m,n}$ satisfies Assumption~\ref{foliation condition} and $\frac{|\partial G|(|\partial G|+1)}{2}\geq |G|$.
\end{lemm}

\begin{proof}
Let $1 \leq l\leq \max\limits_{x\in G}d(x,\partial G)-1$.

For $l~mod~4=1$, the set $N_l$ in $H_{m,n}$ is
\begin{align*}
N_l=&\left\{\frac{l-1}{4}\vec{v}_1 + (\frac{l-1}{2}+k)\vec{v}_2+\omega^3 \mid 0\leq k\leq n-l\right\}:=\{x_0,x_1,\cdots,x_{n-l}\}\\
&\cup \left\{\frac{3(m-1)}{2}-\frac{l-1}{4}\vec{v}_1+ (\frac{l-1}{2}+k)\vec{v}_2+\omega^0 \mid 0\leq k\leq n-l \right\}:=\{y_0,y_1,\cdots,y_{n-l}\}\\
&\cup \left \{(\frac{l-1}{4}+j)\vec{v}_1 + \frac{l-1}{2}\vec{v}_2+\omega^4, (\frac{l-1}{4}+j)\vec{v}_1 + \frac{l-1}{2}\vec{v}_2+\omega^5 \mid 0\leq j\leq \frac{m-l}{2} \right\} \\
&\cup \left\{(\frac{l-1}{4}+j)\vec{v}_1 +(n-\frac{l+1}{2})\vec{v}_2+\omega^2, (\frac{l-1}{4}+j)\vec{v}_1 +(n-\frac{l+1}{2})\vec{v}_2+\omega^1 \mid 0\leq j\leq \frac{m-l}{2}  \right\}.
\end{align*}
For each vertex respectively in the above four subsets of $N_l$, there exists a boundary vertex closest to it in $(\partial G)_L, (\partial G)_R, (\partial G)_B, (\partial G)_T$ respectively.
The decomposition of $N_l$ is
\begin{align*}
    N_l^1=& N_l\setminus \{x_1,x_2,\cdots,x_{n-l-1},y_1,y_2,\cdots,y_{n-l-1}\},\\
    N_l^{p+1}=& \{x_p,x_{n-l-p},y_p,y_{n-l-p} \mid p\leq m-1-p
\},\quad p\in \mathbb{N}_+.
\end{align*}

For $l~mod~4=2$, the set $N_l$ in $H_{m,n}$ is
\begin{align*}
N_l=&\left\{\frac{l-2}{4}\vec{v}_1 + (\frac{l-2}{2}+k)\vec{v}_2+\omega^2 \mid 0\leq k\leq n-l \right\}\\
&\cup \left\{\frac{3(m-1)}{2}-\frac{l-2}{4}\vec{v}_1+ (\frac{l-2}{2}+k)\vec{v}_2+\omega \mid 0\leq k\leq n-l \right\}\\
&\cup \left\{(\frac{l-2}{4}+j)\vec{v}_1 + \frac{l-2}{2}\vec{v}_2+\omega^0, (\frac{l-2}{4}+j)\vec{v}_1 + \frac{l-2}{2}\vec{v}_2+2\omega^0 \mid 0\leq j\leq \frac{m-l-1}{2} \right\}\\
&\cup \left\{(\frac{l-2}{4}+j)\vec{v}_1 +(n-\frac{l}{2})\vec{v}_2+\omega^0, (\frac{l-2}{4}+j)\vec{v}_1 +(n-\frac{l}{2})\vec{v}_2+2\omega^0 \mid 0\leq j\leq \frac{m-l-1}{2}  \right\}.
\end{align*}
For each vertex respectively in the above four subsets of $N_l$, there exists a boundary vertex closest to it in $(\partial G)_L, (\partial G)_R, (\partial G)_B, (\partial G)_T$ respectively.
The decomposition of $N_l$ is $N_l=N_l^1$.

For $l~mod~4=3$, the set $N_l$ in $H_{m,n}$ is
\begin{align*}
N_l=&\left\{\frac{l-3}{4}\vec{v}_1 + (\frac{l-3}{2}+k)\vec{v}_2+\omega \mid 1\leq k\leq n-l \right\}:=\{\gamma_1,\cdots,\gamma_{n-l}\}\\
&\cup  \left\{\frac{3(m-1)}{2}-\frac{l-3}{4}\vec{v}_1+ (\frac{l-3}{2}+k)\vec{v}_2+\omega^2 \mid 1\leq k\leq n-l \right \}:=\{\tau_1,\cdots,\tau_{n-l}\} \\
&\cup \left \{(\frac{l-3}{4}+j)\vec{v}_1 + \frac{l-3}{2}\vec{v}_2+\omega, (\frac{l-3}{4}+j)\vec{v}_1 + \frac{l-3}{2}\vec{v}_2+\omega+2 \mid 0\leq j\leq \frac{m-l}{2} \right\}\\
&\cup \left\{(\frac{l-3}{4}+j)\vec{v}_1 +(n-\frac{l-1}{2})\vec{v}_2+\omega^5, (\frac{l-3}{4}+j)\vec{v}_1 +(n-\frac{l-1}{2})\vec{v}_2+\omega^5+2
 \mid 0\leq j\leq \frac{m-l}{2}  \right\}.
\end{align*}
For each vertex respectively in the above four subsets of $N_l$, there exists a boundary vertex closest to it in $(\partial G)_L, (\partial G)_R, (\partial G)_B, (\partial G)_T$, respectively.
The decomposition of $N_l$ is
\begin{align*}
    N_l^1=& N_l\setminus \{\gamma_1,\cdots,\gamma_{n-l},\tau_1,\cdots,\tau_{n-l}\},\\
    N_l^{p+1}=& \{\gamma_p,\gamma_{n-l-p},\tau_p,\tau_{n-l-p} \mid p\leq m-1-p
\},\quad p\in \mathbb{N}_+.
\end{align*}

For $l~mod~4=0$, the set $N_l$ in $H_{m,n}$ is
\begin{align*}
N_l=&\left\{\frac{l-4}{4}\vec{v}_1 + (\frac{l-4}{2}+k)\vec{v}_2+2\omega \mid 1\leq k\leq n-l\right\}\\
&\cup \left\{\frac{3(m-1)}{2}-\frac{l-4}{4}\vec{v}_1+ (\frac{l-4}{2}+k)\vec{v}_2+2\omega^2 \mid 1\leq k\leq n-l \right\}\\
&\cup \left\{(\frac{l-4}{4}+j)\vec{v}_1 + \frac{l-4}{2}\vec{v}_2+2\omega, (\frac{l-4}{4}+j)\vec{v}_1 + \frac{l-4}{2}\vec{v}_2+2\omega+1 \mid 0\leq j\leq \frac{m-l+1}{2} \right\}\\
&\cup \left\{(\frac{l-4}{4}+j)\vec{v}_1 +(n-\frac{l}{2}+1)\vec{v}_2+2\omega^5, (\frac{l-4}{4}+j)\vec{v}_1 +(n-\frac{l}{2}+1)\vec{v}_2+2\omega^5+1 \right.\\
&\left.
\mid 0\leq j\leq \frac{m-l+1}{2}  \right\}.
\end{align*}
For each vertex respectively in the above four subsets of $N_l$, there exists a boundary vertex closest to it in $(\partial G)_L, (\partial G)_R, (\partial G)_B, (\partial G)_T$ respectively.
The decomposition of $N_l$ is $N_l=N_l^1$.

If $l=\max\limits_{x\in G}d(x,\partial G)$, take $N_l=N_l^1$.

To show the relation between the boundary and interior vertices, simply notice that $|\partial G|=2(m+n+1)$ and $|G|=(m+1)(2n+1)$, thus $\frac{|\partial G|(|\partial G|+1)}{2} =(m+n+1)(2m+2n+3)\geq (m+1)(2n+1) = |G|$.
\end{proof}

\section{Proof of Theorem~\ref{Neubspecd_to_NDmap}}\label{Calculate the ND map from the Neumann boundary spectral data}

We prove Theorem \ref{Neubspecd_to_NDmap}, which gives an explicit formula to represent the ND map in terms of the Neumann boundary spectral data.

\medskip
The following Green's formula is proved in \cite[Lemma 2.1]{MR4620352}.
\begin{lemm}\label{Green's formula on graph}
\rm(Green's formula).
Let $u_1,~u_2:~\bar{G}\rightarrow\mathbb{R}$ be two real-valued functions on $\bar{G}$. Then 
\begin{align*}
(u_1,\Delta_G u_2)_G-(u_2,\Delta_G u_1)_G = (u_2,\partial_{\nu}u_1)_{\partial G}-(u_1,\partial_{\nu}u_2)_{\partial G}.
\end{align*}
\end{lemm}

For each $j=1,2,\dots, |\partial G|$, the scalar orthogonal projection of the wave time solution $u^f(t)$ onto the Neumann eigenfunction $\phi_j$ is denoted by 
\begin{align*}
a_j(t):=(u^f(t),\phi_j)_G = \sum_{x\in G}\mu_x u^f(t,x) \phi_j(x), \quad t=0,1,2,\cdots.
\end{align*}
These scalar orthogonal projections can be explicitly computed from the Neumann boundary spectral data as follows.

\begin{lemm}
For $t=1,2,3,\dots$, we have
\begin{align}\label{a_j(t)}
a_j(t)=c_t(f(1),\phi_j)_{\partial G}+c_{t-1}(f(2),\phi_j)_{\partial G}+\cdots+c_{2}(f(t-1),\phi_j)_{\partial G}+c_{1}(f(t),\phi_j)_{\partial G},
\end{align}
where the constants $c_k$ are defined recursively by $c_1=0$, $c_2=-1$, and $c_k=(2-\lambda_j)c_{k-1}-c_{k-2}$ for $k\geq3$. 
\end{lemm}

\begin{proof}
Apply $D_{tt}$ to $a_j$ to get
\begin{align*}
D_{tt}a_j(t) & = ( D_{tt} u^f(t), \phi_j)_G \\
 & = ( \Delta_G  u^f(t), \phi_j)_G \\
 & = (u^f(t), \Delta_G \phi_j)_G - (f(t), \phi_j)_{\partial G} \\
 & =  -(u^f(t),\lambda_j \phi_j)_G - (f(t), \phi_j)_{\partial G} \\
 & = -\lambda_j a_j(t) - (f(t), \phi_j)_{\partial G},
\end{align*}
where the second equality follows from the wave equation, the third equality  is derived from  the Green's formula in Lemma \ref{Green's formula on graph} with $\partial_\nu\phi_j=0$, and the  final  equality from the definition of $a_j(t)$. 
This is a finite difference equation for $a_j(t)$, which, using the definition of $D_{tt}$, 
can be defined by the inductive relation
\begin{align*}
a_j(t+1) = (2-\lambda_j) a_j(t)-a_j(t-1)-(f(t),\phi_j)_{\partial G}.
\end{align*}
The initial conditions are  given by
\begin{align*}
a_j(0) = (u^f(0),\phi_j)_G = 0, \quad a_j(1) = (u^f(1),\phi_j)_G = 0,
\end{align*}
which are derived from the initial conditions of the function $u^f$. 
We now prove the validity of  formula~\eqref{a_j(t)} by induction.

The base case $t=1$ is true since $c_1=0$. For the inductive step,  suppose that the formula~\eqref{a_j(t)} has been proved for all positive integers less than or equal to $t$, then
\begin{align*}
a_j(t-1)=& c_{t-1}(f(1),\phi_j)_{\partial G}+ \cdots +(\lambda_j-2)(f(t-3),\phi_j)_{\partial G}
-(f(t-2),\phi_j)_{\partial G},\\
a_j(t)=& c_t(f(1),\phi_j)_{\partial G}+\cdots + (\lambda_j-2)(f(t-2),\phi_j)_{\partial G}-(f(t-1),\phi_j)_{\partial G}.
\end{align*}
Insert these representations into the inductive relation to get
\begin{align*}
    a_j(t+1) & = (2-\lambda_j) a_j(t)-a_j(t-1)-(f(t),\phi_j)_{\partial G} \\
    & = ((2-\lambda_j)c_t-c_{t-1})(f(1),\phi_j)_{\partial G}+\cdots +(-(2-\lambda_j)^2+1)(f(t-2),\phi_j)_{\partial G}\\
& \quad -(2-\lambda_j)(f(t-1),\phi_j)_{\partial G}-(f(t),\phi_j)_{\partial G}.
\end{align*}
This completes the proof.
\end{proof}

Now, we prove Theorem~\ref{Neubspecd_to_NDmap}.

\begin{proof}[Proof of Theorem~\ref{Neubspecd_to_NDmap}]
As $\{\phi_j(x),x\in G\}^{|G|}_{j=1}$ forms an orthonormal basis of $l^2(G)$, we can write
\begin{align*}
u^f(t,x) = \sum^{|G|}_{j=1} (u^f(t),\phi_j)_G\phi_j(x) = \sum^{|G|}_{j=1} a_j(t)\phi_j(x),\quad x\in G.
\end{align*}
As each boundary vertex $z\in\partial G$ is connected to a unique interior vertex $x\in G$,
 using  the definition of $\partial_\nu u(z)$ in~\eqref{def_Neumann boundary value}, we get
\begin{align*}
f(t,z) = \partial_\nu u^f(t,z) = \frac{1}{\mu_z} w(x,z) (u^f(t,x) - u^f(t,z)).
\end{align*}
Solving for the Dirichlet data of the wave solution from this relation, we can obtain
\begin{align*}
\Lambda_\mu f(t,z) = u^f(t,z)=u^f(t,x)-\frac{\mu_zf(t,z)}{w(x,z)}= \sum^{|G|}_{j=1} a_j(t)\phi_j(z)-\frac{\mu_zf(t,z)}{w(x,z)},\quad x\sim z, \qquad t\geq 1.
\end{align*}
Theorem \ref{Neubspecd_to_NDmap} is proved by substituting the expression for $a_j(t)$ into the summation given in equation \eqref{a_j(t)}.
\end{proof}

\section{The reconstruction procedure}\label{Prove the reconstruction procedure}

In this section, we introduce a discrete version of the boundary control method. Additionally,  using this method, we will demonstrate the reconstruction procedure for the interior vertex weight.

\subsection{Calculating Inner Products of Waves.}
When $x$ serves as the spatial component of the function $u(t,x)$, we  simply denote $u(t,x)$ as $u(t)$.
For a  function $u(t,x)$, we introduce  the time reversal operator:
\begin{align*}
&\mathscr{R}: l^2\left(\{1,\cdots,T-1\}\times{\partial G}\right)\longmapsto l^2(\{1,\cdots,T-1\}\times{\partial G}),\\
&\mathscr{R} u(t):=u(T-t),\qquad t\in\{1,\cdots,T-1\}.
\end{align*}
Similarly, we introduce another operator:
\begin{align*}
&\mathscr{J}: l^2(\{0,1,\cdots,2T\}\times{\partial G})\longmapsto l^2(\{1,\cdots,T-1\}\times{\partial G}),\\
&\mathscr{J} u(t):=\sum\limits_{j=0}^{T-t-1}u(t+1+2j),\qquad t\in\{1,\cdots,T-1\}.
\end{align*}

Define  $P_T: l^2(\{0,\cdots,2T\}\times{\partial G})\mapsto  l^2(\{1,\cdots,T-1\}\times{\partial G})$ as  the truncation operator, while the adjoint operator $P^*_{T}: l^2(\{1,\cdots,T-1\}\times{\partial G})\mapsto  l^2(\{0,\cdots,2T\}\times{\partial G})$ is an extension operator. The  values of  $P^*_{T}$ on $\{0,T,T+1,\cdots,2T\}\times\partial G$ are extended by zero.
Define  the Dirichlet trace operator by $\tau_Du(t)=u(t)|_{\partial G}$ and the Neumann trace operator by $\tau_Nu(t)=\partial_\nu u(t)|_{\partial G}$, respectively.

For functions $v_1, v_2: \{1,\cdots,T-1\}\times \partial G\longrightarrow \mathbb{R}$, define the inner product on the boundary $l^2$-space $l^2(\{1,\cdots,T-1\}\times\partial G)$ as follows:
\begin{align*}
(v_1,v_2)_{\{1,\cdots,T-1\}\times\partial G}:= \sum^{T-1}_{t=1} (v_1(t), v_2(t))_{\partial G} = \sum\limits_{z\in \partial G}\mu_z\sum\limits_{t=1}^{T-1}v_1(t,z)v_2(t,z).
\end{align*}
Here,  for $z\in \partial G$, the values  $\mu_z$ are known.

The following is a discrete counterpart of the generalized Blagovescenskii identity ({\cite{MR4369044}). The original Blagovescenskii identity is proved in (\cite{Blagoveshchenskii1967}).

\begin{lemm}\label{(u^f(T),v(T))}
Let $u^f$ be the solution of \eqref{eq:ibvp1}, and let  $v\in l^2(\{0,\cdots,2T\}\times\bar{G})$ be the solution of the equation
\begin{align*}
D_{tt} v(t,x) - \Delta_G v(t,x) = 0,\quad (t,x)\in\{1,2,\cdots, 2T-1\}\times G,
\end{align*}
then
\begin{align*}
(u^f(T),v(T))_G=(P_{T}(\Lambda_\mu f), \mathscr J\tau_Nv)_{\{1,\cdots,T-1\}\times{\partial G}}-(P_{T}f, \mathscr J \tau_Dv)_{\{1,\cdots,T-1\}\times\partial G}.
\end{align*}
\end{lemm}

\begin{proof}
Set 
\begin{align*}
I(t,s):=(u^f(t),v(s))_G, \quad t,s \in \mathbb{N},~ t,s \geq 0.
\end{align*}
Using equation \eqref{eq:ibvp1} and the Green's formula,  we get 
\begin{align*}
(D_{tt}-D_{ss})I(t,s)
=&(D_{tt}u^f(t),v(s))_G-(u^f(t),D_{ss}v(s))_G\\
=&(\Delta_Gu^f(t),v(s))_G-(u^f(t),\Delta_Gv(s))_G\\
=&(u^f(t),\partial_\nu v(s))_{\partial G}-(v(s),\partial_\nu u^f(t))_{\partial G}\\
=&(\Lambda_{\mu}f(t),\tau_Nv(s))_{\partial G}-(\tau_Dv(s),f(t))_{\partial G},\quad s, t\geq 1.
\end{align*}
Let us denote the right hand side by $F(t, s)$, i.e.,
\begin{equation*}
F(t,s):=(\Lambda_{\mu}f(t),\tau_Nv(s))_{\partial G}-(\tau_Dv(s),f(t))_{\partial G}. 
\end{equation*}
Using the definition of $D_{tt}$ and $D_{ss}$, the relation can be written as
\begin{align}\label{difference equation}
I(t+1,s)=-I(t-1,s)+I(t,s+1)+I(t,s-1)+F(t,s), \quad  t,s \geq  1
\end{align}
On the other hand, the initial conditions for $u^f$ are $u^f(0,x)=0$ for $x\in\bar{G}$ and $D_tu^f(0,x) = 0$ for $x\in G$. 
Hence,
\begin{equation} 
\begin{aligned}\label{eq:initcond}
I(0,s)= & (u^f(0),v(s))_G = 0,\quad s\in\{0,1,\cdots, 2T\}, \\
I(1,s) = & D_t I(0,s) + I(0,s) = (D_t u^f(0),v(s))_G + I(0,s) = 0, \quad s\in\{0,1,\cdots, 2T\}.
\end{aligned}
\end{equation}
Consequently, we obtain a recursive relationship for $I(t,s)$ with  initial conditions. 
The solution to this recursive relationship is given by
\begin{align}\label{general solution}
I(t,s)=\sum\limits_{i=1}^{t-1}\sum\limits_{j=0}^{t-i-1}F(i,s-t+i+1+2j), \qquad t \geq 2, \quad s\geq 1.
\end{align}
This solution can be proved by induction.  Indeed, when $t=2$, we have from~\eqref{difference equation} and~\eqref{eq:initcond} that
\begin{align*}
I(2,s)=-I(0,s)+I(1,s+1)+I(1,s-1)+F(1,s)=F(1,s),
\end{align*}
which agrees with the solution formula~\eqref{general solution}. This establishes the base case. 

For the inductive step, suppose the solution formula~\eqref{general solution} holds for all $t\leq k$ for some positive integers $k\geq  2$.  Considering   the case $t=k+1$ and  using the recursive relation, we can get
\begin{align*}
I(k+1,s) & = -I(k-1,s)+I(k,s+1)+I(k,s-1)+F(k,s) \\
& = -\sum\limits_{i=1}^{k-2}\sum\limits_{j=0}^{k-i-2}F(i,s-k+i+2+2j)+\sum\limits_{i=1}^{k-1}\sum\limits_{j=0}^{k-i-1}F(i,s-k+i+2+2j) \\
& +\sum\limits_{i=1}^{k-1}\sum\limits_{j=0}^{k-i-1}F(i,s-k+i+2j)+F(k,s)\\
:= & -I_1+I_2+I_3+F(k,s).
\end{align*}
Notice that
\begin{align*}
    I_2 & = I_1 + \sum^{k-2}_{i=1} F(i,s-k+i+2+2j)|_{j=k-i-1} + F(i,s-k+i+2+2j)|_{i=k-1, j=0} \\
    & = I_1 + \sum^{k-2}_{i=1} F(i,s+k-i) + F(k-1,s+1) \\
    & = I_1 + \sum^{k-1}_{i=1} F(i,s+k-i),
\end{align*}
hence,
\begin{align*}
    I(k+1,s) & = -I_1+I_2+I_3+F(k,s) \\
    & = I_3 + \sum^{k-1}_{i=1} F(i,s+k-i) + F(k,s) \\
    & = \sum\limits_{i=1}^{k-1}\sum\limits_{j=0}^{k-i-1}F(i,s-k+i+2j) + \sum^{k-1}_{i=1} F(i,s+k-i) + F(k,s) \\
    & = \sum\limits_{i=1}^{k-1} \left( \sum\limits_{j=0}^{k-i-1} F(i,s-k+i+2j) + F(i,s-k+i+2j)|_{j=k-i} \right) + F(k,s) \\
    & = \sum\limits_{i=1}^{k-1} \left( \sum\limits_{j=0}^{k-i}F(i,s-k+i+2j) \right) + F(i,s-k+i+2j)|_{i=k,j=0} \\
    & = \sum\limits_{i=1}^{k} \sum\limits_{j=0}^{k-i}F(i,s-k+i+2j),
\end{align*}
which agrees with the solution formula~\eqref{general solution} with $t=k+1$. This completes the induction.

Finally,  by substituting  $t=s=T$ into~\eqref{general solution},  we obtain the following expression
\begin{align*}
\begin{split}
(u^f(T),v(T))_G
=&I(T,T)\\
=&\sum\limits_{i=1}^{T-1}\sum\limits_{j=0}^{T-i-1}F(i,i+1+2j) \\
=&\sum\limits_{i=1}^{T-1}\sum\limits_{j=0}^{T-i-1}\left[ (\Lambda_{\mu}f(i),\tau_Nv(i+1+2j))_{\partial G}-(\tau_Dv(i+1+2j),f(i))_{\partial G}\right]\\
=&\sum\limits_{i=1}^{T-1}\left[(\Lambda_{\mu}f(i),\sum\limits_{j=0}^{T-i-1}\tau_Nv(i+1+2j))_{\partial G}-(f(i),\sum\limits_{j=0}^{T-i-1}\tau_Dv(i+1+2j))_{\partial G}\right]\\
=&\sum_{i=1}^{T-1} \left[(\Lambda_{\mu}f(i),\mathscr J\tau_Nv(i))_{\partial G}-(f(i),\mathscr J\tau_D v(i))_{\partial G}\right] \\
=&(P_T(\Lambda_{\mu}f), \mathscr J\tau_Nv)_{\{1,\cdots,T-1\}\times\partial G}-(P_Tf, \mathscr J\tau_Dv)_{\{1,\cdots,T-1\}\times\partial G}.
\end{split}
\end{align*}
\end{proof}

Lemma \ref{(u^f(T),v(T))} shows that $(u^f(T),v(T))_G$ can be expressed by the ND map $\Lambda_\mu$.
Next, we approximate the wave on interior vertices at time $T$.

Define a linear operator  $W: l^2(\{1,\cdots,T-1\}\times\partial G)\longmapsto l^2(G)$ which is $h\longmapsto u^{P^*_Th}(T,x),~x\in G$.
It maps the Neumann boundary value to the solution of equation \eqref{eq:ibvp1} at time $t=T$ and $x\in G$.
Denote its adjoint operator by $W^*$.

\subsection{Calculating \texorpdfstring{$W^*W$}{}}

Let $\Lambda_{\mu,T}$ be the restricted ND map, for $h\in l^2({\{1,\cdots,T-1\}\times\partial G})$,
\begin{align*}
\Lambda_{\mu,T} h:= P_T(\Lambda_\mu P_T^*h)=u^{P_T^*h}|_{\{1,\cdots,T-1\}\times\partial G}.
\end{align*}

The Blagove\v{s}\v{c}enski\u{\i}'s identity was proposed in \cite{Blagoveshchenskii1967}. Let  us give the 
following discrete version.
\begin{prop} \label{thm:WstarW}
Let $h_1, h_2\in l^2(\{1,2,\cdots,T\}\times\partial G)$, and 
 $u^{P_T^*h_1}$, $u^{P_T^*h_2}$  be the solution of \eqref{eq:ibvp1} with Neumann boundary values $P_T^*{h_1}, P_T^*{h_2}\in l^2(\{0,1,\cdots,2T\}\times\partial G)$ respectively. 
 Then, we can obtain
\begin{align*}
(u^{P_T^*h_1}(T),u^{P_T^*h_2}(T))_G&=(Wh_1,Wh_2)_G=(h_1,W^*Wh_2)_{\{1,2,\cdots,T-1\}\times\partial G}\\
&=(h_1,(\mathscr{R}\Lambda_{\mu,T}\mathscr{R}\mathscr{J}P^*_T-\mathscr{J}\Lambda_{\mu}P_T^*)h_2)_{\{1,2,\cdots,T-1\}\times\partial G}.
\end{align*}
\end{prop}

\begin{proof}
As $\tau_Du^{P_T^*h_2}=\Lambda_{\mu}P_T^*h_2$, $\tau_Nu^{P_T^*h_2}=P_T^*h_2$, by Lemma \ref{(u^f(T),v(T))} we have
\begin{align*}
(u^{P_T^*h_1}(T),u^{P_T^*h_2}(T))_G
&=(P_T(\Lambda_{\mu}P_T^*h_1),\mathscr{J}P^*_Th_2)_{\{1,2,\cdots,T-1\}\times\partial G}-(h_1,\mathscr{J}\Lambda_{\mu}P^*_Th_2)_{\{1,2,\cdots,T-1\}\times\partial G}\\
&=(\Lambda_{\mu,T}h_1,\mathscr{J}P^*_Th_2)_{\{1,2,\cdots,T-1\}\times\partial G}-(h_1,\mathscr{J}\Lambda_{\mu}P^*_Th_2)_{\{1,2,\cdots,T-1\}\times\partial G}\\
&=(h_1,\Lambda_{\mu,T}^*\mathscr{J}P^*_Th_2)_{\{1,2,\cdots,T-1\}\times\partial G}-(h_1,\mathscr{J}\Lambda_{\mu}P^*_Th_2)_{\{1,2,\cdots,T-1\}\times\partial G}.
\end{align*}
The adjoint is $\Lambda_{\mu,T}^*=\mathscr{R}\Lambda_{\mu,T}\mathscr{R}$ ( see Appendix \ref{appendixB} for the calculation). Thus,
\begin{align*}
(u^{P_T^*h_1}(T),u^{P_T^*h_2}(T))_G=(h_1,(\mathscr{R}\Lambda_{\mu,T}\mathscr{R}\mathscr{J}P_T^*-\mathscr{J}\Lambda_{\mu}P_T^*)h_2)_{\{1,2,\cdots,T-1\}\times\partial G}.
\end{align*}
On the other hand, the definition of the operator $W$ implies
\begin{align*}
(u^{P_T^*h_1}(T),u^{P_T^*h_2}(T))_G=(Wh_1,Wh_2)_G=(h_1,W^*Wh_2)_{\{1,2,\cdots,T-1\}\times\partial G}.
\end{align*}
Since  $h_1, h_2$ are arbitrary, we can conclude that
\begin{align}\label{WstarWformula}
W^*W=\mathscr{R}\Lambda_{\mu,T}\mathscr{R}\mathscr{J}P_T^*-\mathscr{J}\Lambda_{\mu}P_T^*.
\end{align}
\end{proof}

The next proposition presents  an explicit formula for computing the action of the operator $W^*$ on any harmonic function $\varphi$.

\begin{prop} \label{thm:Wharm}
For $h\in l^2(\{1,2,\cdots,T\}\times\partial G)$, let $u^{P_T^*h}$ be the solution of equation ~\eqref{eq:ibvp1} with the Neumann boundary value $P_T^*h\in l^2(\{0,1,\cdots,2T\}\times\partial G)$.
If $\varphi(x)$ is an arbitrary harmonic function, then
\begin{align*}
(u^{P_T^*h}(T),\varphi)_G
=(h,(\mathscr{R}\Lambda_{\mu,T}\mathscr{R}\mathscr{J}\tau_N-\mathscr{J}\tau_D)\varphi)_{\{1,2,\cdots,T-1\}\times\partial G}.
\end{align*}
Therefore,
\begin{align}\label{Wstarformula}
W^*\varphi=(\mathscr{R}\Lambda_{\mu,T}\mathscr{R}\mathscr{J}\tau_N-\mathscr{J}\tau_D)\varphi.
\end{align}
\end{prop}

\begin{proof}
We take $f=P_T^*h$, $v=\varphi$ in Lemma \ref{(u^f(T),v(T))} and use the fact that $\Lambda_{\mu,T}^* = \mathscr{R}\Lambda_{\mu,T}\mathscr{R}$ (see Appendix \ref{appendixB}) to get

\begin{align*}
(u^{P_T^*h}(T),\varphi)_G
&=(P_T(\Lambda_{\mu}P_T^*h),\mathscr{J}\tau_N\varphi)_{\{1,2,\cdots,T-1\}\times\partial G}-(h,\mathscr{J}\tau_D\varphi)_{\{1,2,\cdots,T-1\}\times\partial G}\\
&=(\Lambda_{\mu,T}h,\mathscr{J}\tau_N\varphi)_{\{1,2,\cdots,T-1\}\times\partial G}- (h,\mathscr{J}\tau_D\varphi)_{\{1,2,\cdots,T-1\}\times\partial G}\\
&=(h,\Lambda_{\mu,T}^*\mathscr{J}\tau_N\varphi)_{\{1,2,\cdots,T-1\}\times\partial G}- (h,\mathscr{J}\tau_D\varphi)_{\{1,2,\cdots,T-1\}\times\partial G}\\
&=(h,(\mathscr{R}\Lambda_{\mu,T}\mathscr{R}\mathscr{J}\tau_N-\mathscr{J}\tau_D)\varphi)_{\{1,2,\cdots,T-1\}\times\partial G}.
\end{align*}
On the other hand, the definition of $W$ gives
\begin{align*}
(u^{P_T^*h}(T),\varphi)_G=(Wh,\varphi)_G=(h,W^*\varphi)_{\{1,2,\cdots,T-1\}\times\partial G}.
\end{align*}
As $h$ is arbitrary, we conclude that  the action of $W^*$ on a harmonic function $\varphi(x)$ is
\begin{align*}
W^*\varphi=(\mathscr{R}\Lambda_{\mu,T}\mathscr{R}\mathscr{J}\tau_N-\mathscr{J}\tau_D)\varphi.
\end{align*}
\end{proof}

\subsection{Solving the Boundary Control Equation}

We aim to determine the existence of a function   $h\in l^2(\{1,\cdots,T-1\}\times\partial G)$ such that $Wh=u^{P_T^*h}(T,x)$ holds for  $x\in G$.
In fact, we need to verify that  $W$ is surjective.

\begin{prop} \label{thm:Wsurj}
Suppose the graph satisfies Assumption~\ref{foliation condition}.
If $T\geq \max\limits_{x\in G}d(x,\partial G)$, then $W$ is surjective.
\end{prop}

\begin{proof}
Note that $W$ is a linear operator between finite dimensional vector spaces. It remains to show that its adjoint operator $W^*$ is injective.

Given any $g\in l^2(G)$, we have (see Appendix \ref{appendixB} for the derivation) 
\[ W^*g = v(t,z) \]
with $(t,z)\in\{{1,2,\cdots,T-1\}\times\partial G}$, 
where $v$ satisfies
\begin{align} \label{v W*eq}
\begin{cases}
D_{tt} v(t,x) - \Delta_Gv(t,x) = 0, & (t,x)\in\{1,2,\cdots, T-1\}\times G, \\
v(T,x) = 0, & x\in \bar{G}, \\
D_t v(T-1,x) =  g(x), & x\in G, \\
\partial_\nu v(t,z)  = 0, & (t,z)\in \{0,1,\cdots, T\}\times \partial G.
\end{cases}
\end{align}
Introduce $V(t) := v(T-t)$.  Then $V$ solves
\begin{align*}
\begin{cases}
D_{tt} V(t,x) - \Delta_G V(t,x) = 0, & (t,x)\in\{1,2,\cdots, T-1\}\times G, \\
V(0,x) = 0, & x\in \bar{G}, \\
D_t V(0,x) =  -g(x), & x\in G, \\
\partial_\nu V(t,z)  = 0, & (t,z)\in \{0,1,\cdots, T\}\times \partial G.
\end{cases}
\end{align*}
Let $V_{odd}(t,x)$ be the odd extension of $V(t,x)$ with respect to $t$, that is
\begin{align*}
V_{odd}(t)=
\begin{cases}
-V(t,x), & t\in \{-T,-T+1,\cdots,-1\},\\
V(t,x), & t\in \{0,1,\cdots,T\}.
\end{cases}
\end{align*}
By construction, 
$V_{odd}(t,x)$ clearly satisfies the wave equation for $t>0$ and $t<0$. For $t=0$, we use $V_{odd}(0,x)=0$ to get
\begin{align*}
D_{tt} V_{odd}(0,x) - \Delta_G V_{odd}(0,x) = & V_{odd}(1,x)-2V_{odd}(0,x)+V_{odd}(-1,x)-\Delta_G V_{odd}(0,x)\\
=& V(1,x)-V(1,x)\\
=& 0.
\end{align*}
Therefore, $V_{odd}(t,x)$ satisfies
\begin{align*}
\begin{cases}
D_{tt} V_{odd}(t,x) - \Delta_G V_{odd}(t,x) = 0, & (t,x)\in\{-T+1,\cdots, T-1\}\times G,\\
V_{odd}(0,x) = 0, & x\in \bar{G},\\
D_tV_{odd}(0,x) = -g(x), & x\in G,\\
\partial_\nu V_{odd}(t,z)  = 0, & (t,z)\in \{-T,\cdots, T\}\times \partial G.
\end{cases}
\end{align*}
If $W^*g(x) = 0$ for $x\in G$, then $V_{odd}(t,z)=0$ for $(t,z)\in\{-T,\cdots,T\}\times\partial G$.
By the  unique continuation property in Theorem~\ref{unique_continuation_principle}, we have 
\begin{align*}
D_tV(0,x) = D_tV_{odd}(0,x) =-g(x)=0
\end{align*}
for every $x\in G$ when $T\geq \max\limits_{x\in G}d(x,\partial G)$.
Therefore,  $W^*$ is injective, which implies that $W$ is surjective.
\end{proof}

\begin{prop}\label{explicit h0}
Suppose the graph satisfies Assumption~\ref{foliation condition} and $T\geq \max\limits_{x\in G}d(x,\partial G)$. For any harmonic function $\psi$, the boundary Neumann data given by 
\begin{equation} \label{h0formula}
h_0 = (W^*W)^\dagger W^* \psi
\end{equation}
satisfies $u^{P_T^*h_0}(T,x)=\psi(x)$ for $x\in G$. Here ${\cdot}^\dagger$ denotes the pseudo-inverse. 

\end{prop}

\begin{proof}
If $T\geq \max\limits_{x\in G}d(x,\partial G)$, we know that $W: l^2({\{1,\cdots, T-1\}\times \partial G})\longmapsto l^2(G)$ is surjective by Proposition~\ref{thm:Wsurj}.
Hence, the equation $u^{P_T^*h}(T)=Wh=\psi$ admits solutions. It remains to prove the explicit formula~\eqref{h0formula}.

Define a quadratic functional $\mathcal{F}$ by
\begin{align*}
\mathcal{F}(h):=\|u^{P_T^*h}(T)-\psi\|_G^2&=\|Wh-\psi\|_G^2\\
&=(Wh,Wh)_G-2(Wh,\psi)_G+\|\psi\|_G^2\\
&=(h,W^*Wh)_{\{1,2,\cdots,T-1\}\times\partial G}-2(h,W^*\psi)_{\{1,2,\cdots,T-1\}\times\partial G}+\|\psi\|_G^2.
\end{align*}
The gradient and Hessian matrix of $\mathcal{F}$ are 
\begin{align*}
\mathcal{F}'(h)=2W^*Wh-2W^*\psi, \quad 
\mathcal{F}''(h)=2W^*W.
\end{align*}
Since  the Hessian matrix is positive semi-definite, the function  $\mathcal{F}$ is convex. Consequently, a local minimum of $\mathcal{F}$ is also a global minimum.
To find a local minimum, we set  $\mathcal{F}'(h)=0$ to obtain the normal equation
\begin{align*}
W^*Wh_0=W^*\psi.
\end{align*}
This is an under-determined linear system, and its  minimum norm least squares solution is given by~\eqref{h0formula}.
\end{proof}

Note that $W^*W$, and $W^*\psi$ can be computed  explicitly from the ND map using Proposition~\ref{thm:WstarW} and Proposition~\ref{thm:Wharm}.
Therefore, the formula~\eqref{h0formula} provides an explicit construction of a boundary control.

\subsection{Constructing \texorpdfstring{$\mu$}{}}

Define 
\begin{align}\label{sapn_space_M}
M = {\rm span}\{\varphi\psi|_G: \varphi, \psi \in l^2(\bar{G}), \Delta_G \varphi(x) = \Delta_G \psi(x) = 0, \ x\in G\},
\end{align}
that is, $M$ is the span of all the products of harmonic function on $G$.
Note that as $\mu_x>0$ for each $x\in G$, the concept of harmonic functions is independent of the weight $\mu$,  and so is $M$.

Let us give the  proof of  Theorem \ref{orthogonal_projection_of_mu_onto_M}.

\begin{proof}
Given two harmonic functions $\psi(x)$ and $\varphi(x)$ on the graph, we can find a boundary control $h_0$ such that $u^{P_T^*h_0}(T) = \psi$ by applying  Proposition~\ref{explicit h0}. Consequently, the following identity holds:
\begin{align}\label{solvemu_xformula}
\sum\limits_{x\in G}\mu_x\psi(x)\varphi(x) = (\psi,\varphi)_G = (u^{P_T^*h_0}(T),\varphi)_G = (Wh_0,\varphi)_G = (h_0,W^*\varphi)_{\{1,2,\cdots,T-1\}\times\partial G}.
\end{align}
The right-hand side can be explicitly calculated using Proposition  \ref{thm:Wharm}. 
The left-hand side represents  the inner product of $\mu$ with the product $\psi\varphi$. By varying the harmonic functions $\psi,\varphi$, we can compute  the orthogonal projection of $\mu|_{G}$ onto the space $M$. 
\end{proof}

\section{Uniqueness and Reconstruction} \label{sec:uniqueandrecon}

The proof only reconstructs the orthogonal projection of $\mu\in l^2(G)$ on the subspace $M\subset l^2(G)$ but not $\mu$ itself. General speaking, $M\neq l^2(G)$, see the discussion below. This is in contrast to the well-known fact that the products of (continuum) harmonic functions on a bounded open set $\Omega\subset\mathbb{R}^n$ $(n\geq 2)$ is dense in $L^2(\Omega)$ \cite{MR590275}. However, we can identify some sufficient conditions so that $M = l^2(G)$ for a generic class of edge weight functions.

Let us index all the vertices $x\in \bar{G}$ so that the interior vertices are indexed by $x_1,\dots, x_{|G|}$ and the boundary vertices by $x_{|G|+1}, \dots,x_{|\bar{G}|}$.
Let $\varphi^{(j)}$ solve the boundary value problem 
\begin{equation}
    \Delta_G \varphi^{(j)}(x) = 0 \text{ for } x\in G, \qquad \varphi^{(j)}|_{\partial G} = \delta^{(j)}
\end{equation}
where $\delta^{(j)}$ is a function on $\partial G$ such that 
\begin{align*}
\delta^{(j)}=
\begin{cases}
1 \quad \text{on}~x_{|G|+j},\\
0 \quad \text{on}~\partial G\backslash \{x_{|G|+j}\}.
\end{cases}
\end{align*}
This boundary value problem admits a unique solution, see Lemma~\ref{thm:ellipticBVP}. 
Denote the space 
\begin{equation*}
H := {\rm span}\{\varphi^{(j)}\in l^2(\bar{G}): j = 1,2,\dots,|\partial G|\}.
\end{equation*}

\begin{lemm}
    H is the space of harmonic functions on $\bar{G}$.
\end{lemm}
\begin{proof}
    It is clear that any function in $H$, as a linear combination of harmonic functions, is harmonic. Conversely, suppose $\varphi \in l^2(\bar{G})$ is an arbitrary harmonic function. Define
    $$
    \tilde{\varphi} := \sum^{|\partial G|}_{j=1} \varphi(x_{|G|+j}) \varphi^{(j)} \quad \in H.
    $$
    Then $\tilde{\varphi}$ is a harmonic function and $\tilde{\varphi}|_{\partial G} = \sum^{|\partial G|}_{j=1} \varphi(x_{|G|+j}) \delta^{(j)} = \varphi|_{\partial G}$. We conclude $\varphi = \tilde{\varphi}$ by Lemma~\ref{thm:ellipticBVP}, hence $\varphi \in H$.
\end{proof}

Using the indices, we can vectorize functions on $\bar{G}$ as follows: A function $u\in l^2(\bar{G})$ can be identified with a vector $\vec{u} = (u(x_1), \dots, u(x_{|\bar{G}|}))^T \in\mathbb{R}^{|\bar{G}|}$ via
\begin{equation} \label{eq:iden}
 l^2(\bar{G}) \ni u \leftrightarrow \vec{u}:= \left(
    \begin{aligned}
            \vec{u}_G \\
            \vec{u}_{\partial G} \\
    \end{aligned} \right)
    \in \mathbb{R}^{|G|}\times\mathbb{R}^{|\partial G|}.
\end{equation}
The vectorized space of harmonic functions is
$$
\vec{H} := {\rm span}\{\vec{\varphi}^{(j)}\in \mathbb{R}^{|\bar{G}|}: j = 1,2,\dots,|\partial G|\}.
$$
The vectorized space of products of harmonic functions on $G$ is
\begin{equation}
    \vec{M} := {\rm span} \left\{
    \vec{\varphi}^{(j)}_G \odot \vec{\varphi}^{(k)}_G \in \mathbb{R}^{|G|}: \; j, k = 1,2,\dots,|\partial G|
    \right\}
\end{equation}
where $\odot$ is the Hadamard product between two vectors.

Using the indices, the graph Laplacian $\Delta_G: \bar{G}\rightarrow G$ is identified with a block matrix
\begin{align*}
[\Delta_G] = ([\Delta_{G,G}], [\Delta_{G,\partial G}]) \;\in \mathbb{R}^{|G|\times |\bar{G}|},
\end{align*}
 where $
[\Delta_{G,G}] \in \mathbb{R}^{|G|\times |G|}, \; 
[\Delta_{G,\partial G}] \in \mathbb{R}^{|G|\times |\partial G|}.$
Then $\vec{\varphi} \in\mathbb{R}^{|\bar{G}|}$ is a vectorized harmonic function if and only if $[\Delta_G] \vec{\varphi} = 0$. The discussion along with the rank-nullity relation leads to the following conclusion:
\begin{lemm}\label{rank_of_H}
  \begin{equation*}
    \vec{H} = \ker [\Delta_G] \quad \text{ and } \quad \dim \vec{H} + {\rm rank} [\Delta_G] = |\bar{G}|.
\end{equation*}
\end{lemm}

The discussion in the rest of this section adapts ideas from \cite{MR3507557}. Let us construct a matrix $\mathbf{H}$ using $\vec{\varphi}^{(j)}_G \odot \vec{\varphi}^{(k)}_G$ as columns,  where  $j\leq k$, and $j, k = 1,\dots,|\partial G|$.
It is evident that $\mathbf{H}\in\mathbb{R}^{|G|\times \frac{|\partial G| (|\partial G|+1)}{2}}$ and the range of $\mathbf{H}$ is $\vec{M}$. Moreover, the following three statements are equivalent:
\begin{enumerate}
    \item $M=l^2(G)$.
    \item $\vec{M}=\mathbb{R}^{|G|}$.
    \item ${\rm rank} (\mathbf{H})=|G|$.
\end{enumerate}

\begin{rema} \label{rema:verticenum}
Since  the rank of a matrix cannot exceed the number of columns, a necessary condition for $rank(\mathbf{H})=|G|$ is that $\frac{|\partial G|(|\partial G|+1)}{2} \geq |G|$.
This condition  requires the graph to have sufficient boundary vertices relative to the interior vertices.
\end{rema}

Note that the entries in $\mathbf{H}$ depend on the edge weight function $w_{x,y}\in \mathbb{R}_+^{|\mathcal{E}|}$, since the definition of $\Delta_G$ involves $w_{x,y}$.
We have the following alternatives for ${\rm rank} (\mathbf{H})$ with respect to $w_{x,y}$.

\begin{prop}\label{Prop_three_conditions_of_rec}
If the graph satisfies $\frac{|\partial G|(|\partial G|+1)}{2}\geq |G|$, then exactly one of the following cases occurs:
\begin{enumerate}
    \item there is no edge weight function $w_{x,y}\in \mathbb{R}_+^{|\mathcal{E}|}$ such that ${\rm rank}(\mathbf{H})=|G|$;
    \item ${\rm rank}(\mathbf{H})=|G|$ for all edge weight functions $w_{x,y}\in \mathbb{R}_+^{|\mathcal{E}|}$ except for a set of measure zero.
\end{enumerate}
\end{prop}

\begin{proof}
Let $\beta$ be an arbitrary selection of $|G|$ columns from $\mathbf{H}$. Observe that
$$
{\rm rank}(\mathbf{H}) = |G| \quad \text{ if and only if } \quad \exists~\beta \text{ such that } \det(\mathbf{H}_{:,\beta})\neq 0,
$$
or equivalently,
$$
{\rm rank}(\mathbf{H}) < |G| \quad \text{ if and only if } \quad\det(\mathbf{H}_{:,\beta})=0, \ \ \forall\beta.
$$
We will use the fact that for a fixed $\beta$, $\det(\mathbf{H}_{:,\beta})$ is a real analytic function of $w_{x,y}$, see Lemma~\ref{app:realanalytic}.

If $\det(\mathbf{H}_{:,\beta})$ is the zero function for all $\beta$, that is, if $\det(\mathbf{H}_{:,\beta})\equiv 0$ regardless of $w_{x,y}$ for all $\beta$, then there is no edge weight function such that ${\rm rank}(\mathbf{H}) = |G|$ holds, accounting for Case (1). 
On the other hand, if there exists $\beta$ such that $\det(\mathbf{H}_{:,\beta})$ is not the zero function, that is $\det(\mathbf{H}_{:,\beta})\not\equiv 0$, then it is a non-trivial real analytic function of $w_{x,y}$, hence the zeros
$$
S_\beta := \{w_{x,y}\in \mathbb{R}_+^{|\mathcal{E}|} : \det(\mathbf{H}_{:,\beta})=0\}
$$
form a set of measure zero~\cite{MR4070868}. The collection of edge weight functions that ensure ${\rm rank}(\mathbf{H}) < |G|$ is
\begin{align*}
     & \{w_{x,y}\in \mathbb{R}_+^{|\mathcal{E}|} : {\rm rank}(\mathbf{H}) < |G|\} \\
     = & \{w_{x,y}\in \mathbb{R}_+^{|\mathcal{E}|} : \det(\mathbf{H}_{:,\beta})=0, \; \forall\beta\} \\
      = & \bigcap_\beta S_\beta. 
\end{align*}
This is a finite intersection of sets of measure zero, hence is also of measure zero.

\end{proof}

We remark that Case (1) in Proposition~\ref{Prop_three_conditions_of_rec} can indeed occur. Here is an example.

\begin{exam} 
For the graph in Fig. \ref{fig_weight_counterexample}, $|\partial G|= |G| = 2$. It is easy to see that the graph satisfies Assumption~\ref{foliation condition} and $\frac{|\partial G|(|\partial G|+1)}{2}=3>|G|=2$. 
For $x_1,x_2\in G$ and $z_1,z_2\in \partial G$, matrix $\mathbf{H}_{:,i+(j-1)|\partial G|}=[{\Delta_{G,G}^{-1}}\Delta_{G,\partial G}]_{:,i}\odot[{\Delta_{G,G}^{-1}}\Delta_{G,\partial G}]_{:,j}$, where $1\leq i,j\leq|\partial G|$ and $\mathbf{H}$ is independent on vertex weight. $\mathbf{H}$ is obviously a $|G|\times|\partial G|^2$ matrix.
In this case,
\begin{align*}
\begin{split}
\mathbf{H}=\frac{1}{(w_{x_1,z_1}+w_{x_1,z_2})^2}
\begin{pmatrix}
 w_{x_1,z_1}^2 & w_{x_1,z_1}w_{x_1,z_2} & w_{x_1,z_1}w_{x_1,z_2} & w_{x_1,z_2}^2 \\
 w_{x_1,z_1}^2 & w_{x_1,z_1}w_{x_1,z_2} & w_{x_1,z_1}w_{x_1,z_2} & w_{x_1,z_2}^2
\end{pmatrix},
\end{split}
\end{align*}
and ${\rm rank} (\mathbf{H})=1<|G|=2$.
\end{exam}

\begin{figure}[htbp]
\centering
\includegraphics[scale=0.66]{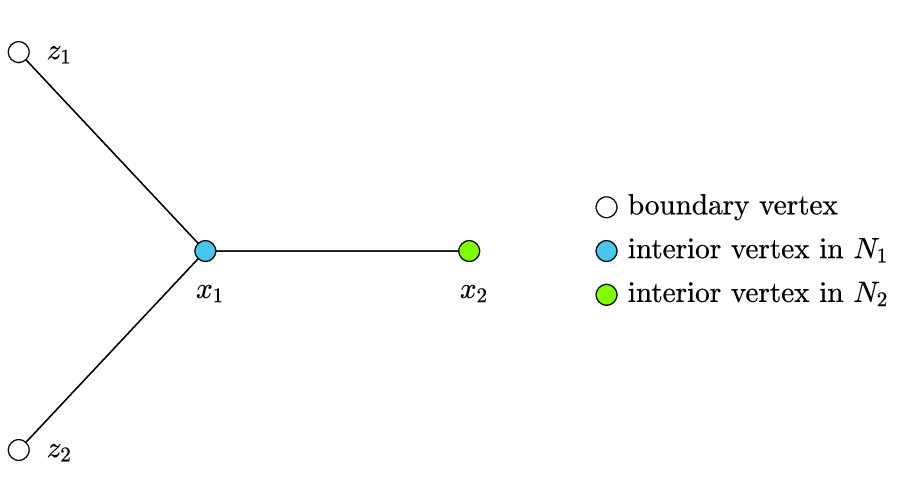}
\caption{\small\small{The interior vertex weight cannot be reconstructed by our method. }}
\label{fig_weight_counterexample}
\end{figure}

\section{Reconstruction Algorithm}\label{sec:Reconstruction Algorithm}

In this section, we implement the reconstruction procedure and validate it using numerical examples. Here, the reconstruction procedure is summarized in Algorithm~\ref{alg:Framwork}.

\begin{algorithm}[htb]
\caption{Reconstruction Algorithm of the interior vertex weight.}
\label{alg:Framwork}
\begin{algorithmic}[1] 
\REQUIRE Time reversal operator $\mathscr{R}$, operator $\mathscr{J}$, truncation operator $P_T$, boundary vertex weight $\mu|_{\partial G}$ and the Neumann boundary spectral data $\{(\lambda_j,\phi_j|_{\partial G})_{j=1}^{|G|}\}$.
\vspace{1ex}

    \STATE Calculate the ND map $\Lambda_{\mu}$ from the Neumann boundary spectral data by \begin{align*}
\Lambda_\mu f(t,z) = u^f(t,z) = \sum^{|G|}_{j=1} \sum\limits_{k=1}^tc_{k}(f(t+1-k),\phi_j)_{\partial G}\phi_j(z)
-\frac{\mu_zf(t,z)}{w(x,z)},\quad x\sim z,
\end{align*}
for $t\geq1$,
where $c_1=0,c_2=-1$, and $c_k=(2-\lambda_j)c_{k-1}-c_{k-2}$ for integer $k\geq3$ (see (\ref{formula_of_Neubsd_to_express_NDmap})).

    \STATE Calculate the operator $W^*W=\mathscr{R}\Lambda_{\mu,T}\mathscr{R}\mathscr{J}P_T^*-\mathscr{J}\Lambda_{\mu}P_T^*$  (see \eqref{WstarWformula}).
    
    \STATE Calculate the operator $W^*$ operating on a harmonic function $\psi$ on a graph 
    \begin{align*}
    W^*\varphi=(\mathscr{R}\Lambda_{\mu,T}\mathscr{R}\mathscr{J}\tau_N-\mathscr{J}\tau_D)\varphi\quad \text{({\rm see } \eqref{Wstarformula})}.
    \end{align*}

    \STATE Solve $h_0$ from $W^*Wh_0=W^*\psi$, where $\psi$ is a harmonic function on the graph  (see (\ref{h0formula})).
    
    \STATE Construct the harmonic functions  $\psi(x)$ and $\varphi(x)$ on the graph to reconstruct $\mu|_G$ from
\begin{align*}
(u^{P_T^*h_0}(T),\varphi)_G=(\psi,\varphi)_G=\sum\limits_{x\in G}\mu_x\psi(x)\varphi(x)=(h_0,W^*\varphi)_{\{1,2,\cdots,T-1\}\times\partial G} \quad 
\text{
({\rm see} \eqref{solvemu_xformula})}.
\end{align*}
\RETURN $\mu|_{G}$.

\vspace{1ex}
\ENSURE  The interior vertex weight, $\mu|_G$;

\end{algorithmic}
\end{algorithm}

To implement the algorithm, we will index the vertices so that functions on graphs can be identified with vectors, and linear operators on graphs can be identified with matrices.
Recall that the vertices of $\bar{G}$ are ordered in the way that the interior vertices are indexed by $x_1,\dots, x_{|G|}$ and the boundary vertices by $x_{|G|+1}, \dots,x_{|\bar{G}|}$. For a spatial function $u\in l^2(\bar{G})$, it is vectorized as in \eqref{eq:iden}.
For a spatiotemporal function $f(t,x)$ with $t=0,1,2,\dots, T$ and $x\in\bar{G}$, we follow the lexicographical order to identify
\begin{equation*}
f \leftrightarrow \vec{f} := (f(0,\cdot), f(1,\cdot), \cdots, f(T,\cdot))^*,
\end{equation*}
where $*$ denotes the adjoint, which is the transpose for a real vector. Using such an ordering, linear operators can be identified with matrices. For instance, the ND map $\Lambda_\mu$ is realized as an ND matrix via the following identification
\begin{equation*}
\Lambda_\mu: l^2(\{0,1,\dots,2T\}\times\partial G) \rightarrow l^2(\{0,1,\dots,2T\}\times\partial G)  \leftrightarrow  [\Lambda_{\mu}]\in \mathbb{R}^{|\partial G|(2T+1) \times |\partial G|(2T+1) }
\end{equation*}
where we use the square parenthesis $[\cdot]$ to indicate matrix representations of linear operators.

The algorithm is implemented in the following steps.

{\bf{Step 1: Assemble the Discrete Neumann-to-Dirichlet matrix.}} 
Given the Neumann boundary spectral data, the ND matrix can be readily calculated  by following the formula presented in ~\eqref{formula_of_Neubsd_to_express_NDmap}.
See Fig. \ref{fig_ND_map} for an example of the ND matrix.

\begin{figure}[htbp]
\centering
\includegraphics[scale=0.15]{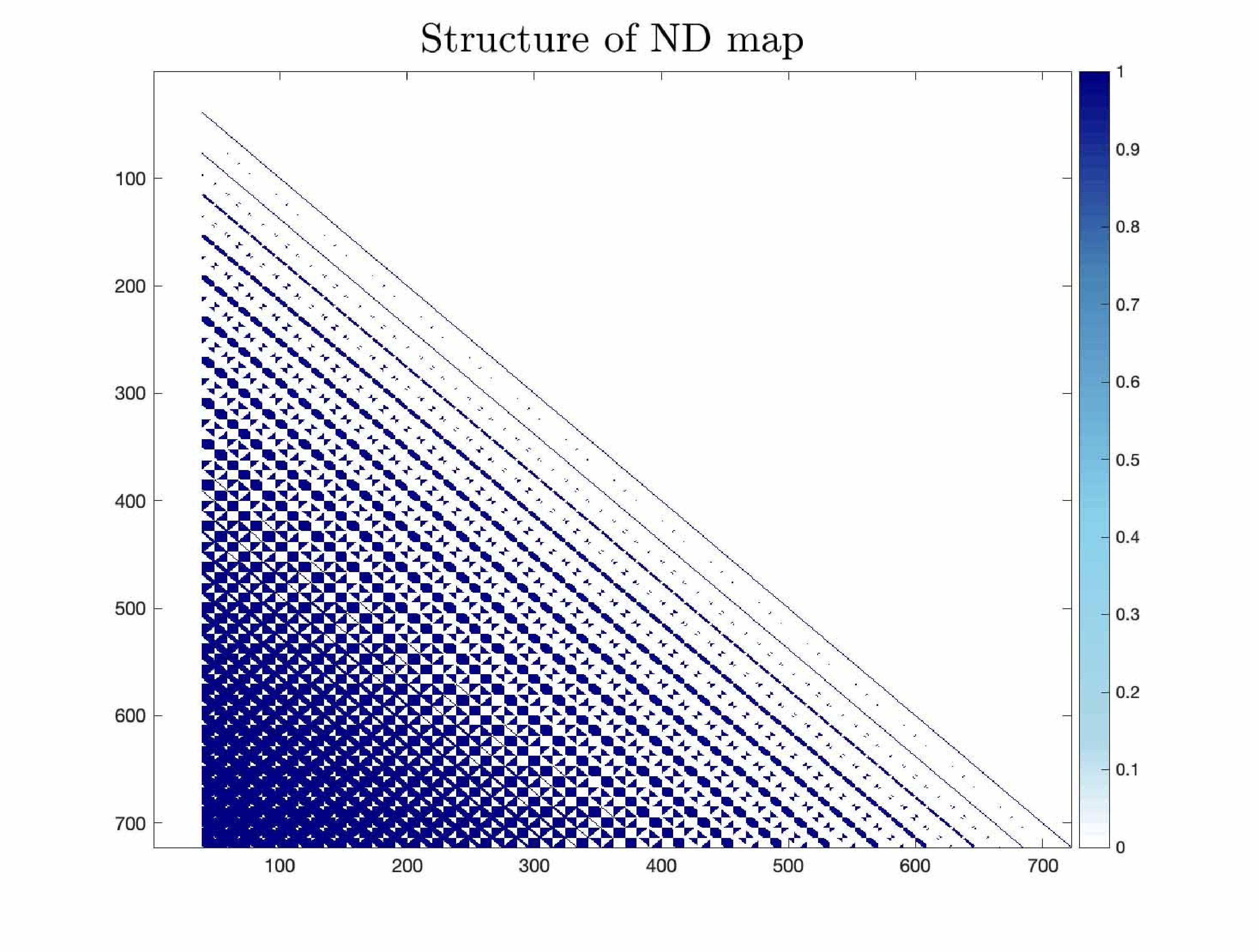}
\caption{\small\small{The structure of the ground-truth ND map $[\Lambda_{\mu}]$ for the graph $R_{m,n}$ when  $m=10, n=9, |\partial G|=38, T=9$. Blank space represents zero element. Blue areas represent nonzero values.}}
\label{fig_ND_map}
\end{figure}

{\bf{Step 2: Calculate the matrix $[W^{*}W]$.}} 
Using the ordering of the vertices, the operators 
\begin{align*}
 \mathscr{R}: & l^2\left(\{1,\cdots,T-1\}\times{\partial G}\right)\longmapsto l^2(\{1,\cdots,T-1\}\times{\partial G}), \\   
 \mathscr{J}: & l^2(\{0,1,\cdots,2T\}\times{\partial G})\longmapsto l^2(\{1,\cdots,T-1\}\times{\partial G}) ,\\
 P_T: & l^2(\{0,\cdots,2T\}\times{\partial G})\longmapsto  l^2(\{1,\cdots,T-1\}\times{\partial G}),
\end{align*}
are represented by the matrices
\begin{equation*}
[\mathscr{R}]\in \mathbb{R}^{|\partial G|(T-1)\times |\partial G|(T-1)}, \quad 
[\mathscr{J}]\in \mathbb{R}^{|\partial G|(T-1)\times |\partial G|(2T+1)}, \quad
[P_T]\in \mathbb{R}^{|\partial G|(T-1)\times |\partial G|(2T+1)}.
\end{equation*}
The matrix representation of the adjoint operator $P^*_T$ is the transpose matrix $[P_T]^*$. 
Following~\eqref{formula_of_Neubsd_to_express_NDmap}, the matrix $[W^{*}W]$ is computed as the matrix product:
\begin{equation}\label{WstarW}
[W^{*}W]=[\mathscr{R}][\Lambda_{\mu,T}][\mathscr{R}][\mathscr{J}][P_T^{*}]-[\mathscr{J}][\Lambda_{\mu}][P_T^{*}] \quad \in \mathbb{R}^{|\partial G|(T-1)\times |\partial G|(T-1) }.
\end{equation}

{\bf{Step 3: Calculate the matrix $[W^{*}]$.}}
Using the ordering of the vertices in $\bar{G}$ and the vectorization~\eqref{eq:iden}, the matrix representations of the Dirichlet trace operator $\tau_D$ and the Neumann trace operator $\tau_N$ are
\begin{equation*}
[\tau_D]=( O\quad I)\ \in  \mathbb{R}^{|\partial G|\times |\bar{G}|}, \quad [\tau_N]=\tau_D\cdot [\Delta_{\bar{G}}]  \in  \mathbb{R}^{|\partial G|\times |\bar{G}|},
\end{equation*}
where $O\in \mathbb{R}^{|\partial G| \times |G|}$ is the zero matrix, $I\in\mathbb{R}^{|\partial G|\times |\partial G|}$ is the identity matrix, and $\cdot$ denotes matrix multiplication.
$[\Delta_{\bar{G}}]\in \mathbb{R}^{|\bar{G}| \times |\bar{G}|}$ is the matrix form of a continuation operator of the graph Laplace $\Delta_{G}$
that its domain of definition is extended from $G$ to $\bar{G}$. 
Following~\eqref{Wstarformula}, we have
\begin{equation}\label{Wstar}
[W^{*}] = [\mathscr{R}][\Lambda_{\mu,T}][\mathscr{R}][\mathscr{J}] (I_{2T+1}\otimes\tau_N) - [\mathscr{J}](I_{2T+1}\otimes\tau_D) \quad \quad\in \mathbb{R}^{|\partial G|(T-1)\times (2T+1)|\bar{G}| },
\end{equation}
where $I_{2T+1} \in\mathbb{R}^{(2T+1)\times (2T+1)}$ is the identity matrix, and $\otimes$ denotes the matrix tensor product. The tensor product is needed as $\tau_D,\tau_N$ are spatial operators while the other operators are spatiotemporal.

{\bf{Step 4: Calculate the Boundary Control $\vec{h}_0$.}} For any harmonic function $\vec{\psi}$, the boundary control $\vec{h}_0 \in \mathbb{R}^{|\partial G|(T-1)\times 1}$ is given by Proposition~\ref{explicit h0}:
\begin{align}\label{h0}
\vec{h}_0 = [W^* W]^\dagger [W^*] (\mathbf{1}_{2T+1}\otimes \vec{\psi}),
\end{align}
where $\mathbf{1}_{2T+1} \in\mathbb{R}^{2T+1}$ denotes the vector of all one's, that is, $\mathbf{1}_{2T+1} = (1,1,\dots, 1)^T$. Again, the tensor product is needed to turn a spatial function into a spatiotemporal one.

{\bf{Step 5: Solve for $\vec{\mu}_{G}$.}} 
Based on~\eqref{solvemu_xformula}, it remains to solve the linear system 
\begin{equation}\label{mu}
 [\vec{\varphi}_G \odot \vec{\psi}_G]^* \vec{\mu}_{G}= \vec{h}^*_0 [W^*](\mathbf{1}_{2T+1}\otimes\vec{\varphi})
\end{equation}
for various vectorized harmonic functions $\vec{\varphi}$ and $\vec{\psi}$. Note that there is a total of ($|\bar{G}|-{\rm rank}(\Delta_G)$) distinct harmonic functions in $G$ by Lemma \ref{rank_of_H}. This results in a linear system, whose reduced row echelon form is calculated using the MATLAB command `rref' in order to obtain $\vec{\mu}_{G}$.

\section{Numerical Experiments}\label{sec:Numerical Expreiments}

In this section, we validate the algorithm using several numerical examples in MATLAB\textsuperscript{TM}. 
We will use two types of discrepancy metrics to measure the difference between quantities. 
The first step of the algorithm requires construction of the ND map, which is represented by a matrix. We will use the 
Frobenius relative norm error ($\mathrm {FRNE}$) 
$$\mathrm {FRNE}=\frac{\lVert [\Lambda_{\mu}]-[\Lambda_\mu']  \rVert_{F} }{\lVert [\Lambda_{\mu}]\rVert_{F} }*100\%$$
to quantify the discrepancy between matrices. Here, $[\Lambda_{\mu}]$ denotes the ground truth ND map and $[\Lambda_\mu']$ denotes the reconstructed ND map based on the algorithm.
For the vertex weight, it is vectorized in the calculation, and the reconstruction accuracy is quantified by the absolute error
 \[\mathrm {Error}:=|\vec{\mu}_x -\vec{\mu}_x'|,\]
as well as the $L_2$-relative norm error ($\mathrm {L_2RNE}$)
\[\mathrm {L_2RNE}:=\frac{\lVert \vec{\mu}_x-\vec{\mu}_x'\rVert_{2} }{\lVert\vec{\mu}_x \rVert_{2} }* 100\%,\]
where $\vec{\mu}_x$ denotes the ground truth and $\vec{\mu}_x'$ denotes the reconstruction.

\subsection{Experiment 1: the graph \texorpdfstring{$R_{m,n}$}{}}

In this experiment, we set the following parameters: $m=10, n=9, |\partial G|=38, |G|=90, T=9$, $w_{x,y}=0.25$, and the ground-truth vertex weight is $\mu_x= \mathrm {deg}(x)$ for all $x \in \bar{G}$.

\textbf{Case 1.1: No Noise.} We implement Algorithm~\ref{alg:Framwork} without noise to validate its efficacy.
The first step of the algorithm assembles the discrete ND map using the Neumann boundary spectral data following~\eqref{formula_of_Neubsd_to_express_NDmap}.  
This can be done with high precision. In fact, let $\Lambda_\mu$ be the ground-truth ND map, and $\Lambda_\mu'$ be the reconstructed ND map using the Neumann boundary spectral data. 
The $\mathrm {FRNE}$ between them is 
$5.9501*10^{-13}\%.$

When solving the equation (\ref{h0}), the matrix $[W^{*}W]$ is ill conditioned, see Fig. \ref{fig_case3_singular} for its singular values. We employ the truncated SVD regularization along with the `lsqminnorm' command in MATLAB to find the minimum norm solution as $\vec{h}_0$.
When solving the  linear equations (\ref{mu}),  we find that ${\rm rank}(\Delta_G)=90$. By Lemma~\ref{rank_of_H}, we conclude the vectorized space of harmonic functions $\vec{H}$ has dimension $|\bar{G}|-{\rm rank}(\Delta_G)=38$. 
In this case, from MATLAB, there are 128 linearly independent vectors of the form $\vec{\varphi}\odot\vec{\psi}$ with $\vec{\varphi}, \vec{\psi}\in \vec{H}$.
Here, we use these $128$ linearly independent vectors as columns to construct the matrix  $[\vec{\varphi}_G \odot \vec{\psi}_G]$ in order to solve~\eqref{mu}. However, the matrix $[\vec{\varphi}_G \odot \vec{\psi}_G]$ is again ill conditioned, as is shown in  Fig. \ref{fig_case3_singular}, so we apply the truncated SVD regularization to find the minimum norm solution.
The reconstruction and the errors are shown in Fig. \ref{fig_case3}.

\begin{figure}[htbp]
\centering
\includegraphics[scale=0.39]{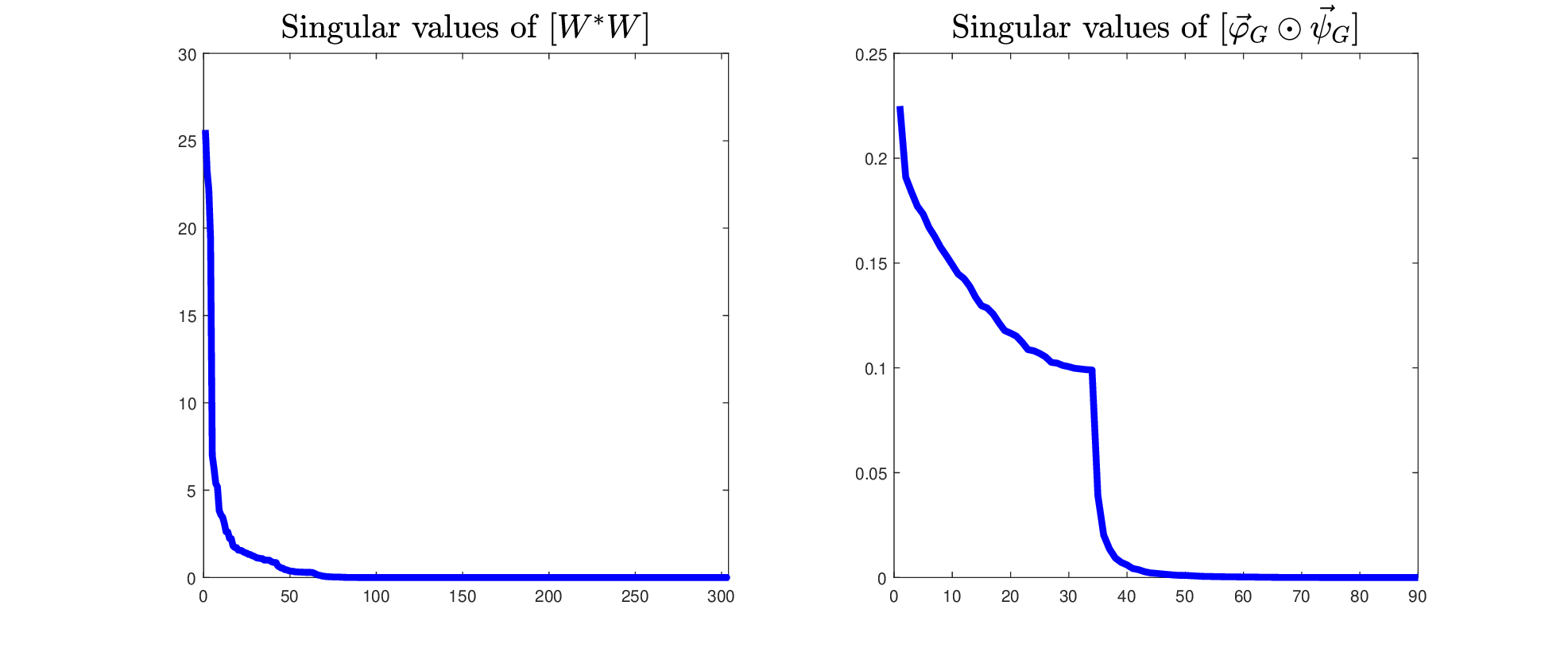}
\caption{\small\small{Experiment 1: The singular values of $[W^{*}W]$ and $[\vec{\varphi}_G \odot \vec{\psi}_G]$. The minimum singular values are $1.3448*10^{-15}$ and $2.5176*10^{-9}$, respectively.}}
\label{fig_case3_singular}
\end{figure}

\begin{figure}[htbp]
\centering
\includegraphics[width=15.5cm,height=4.5cm]{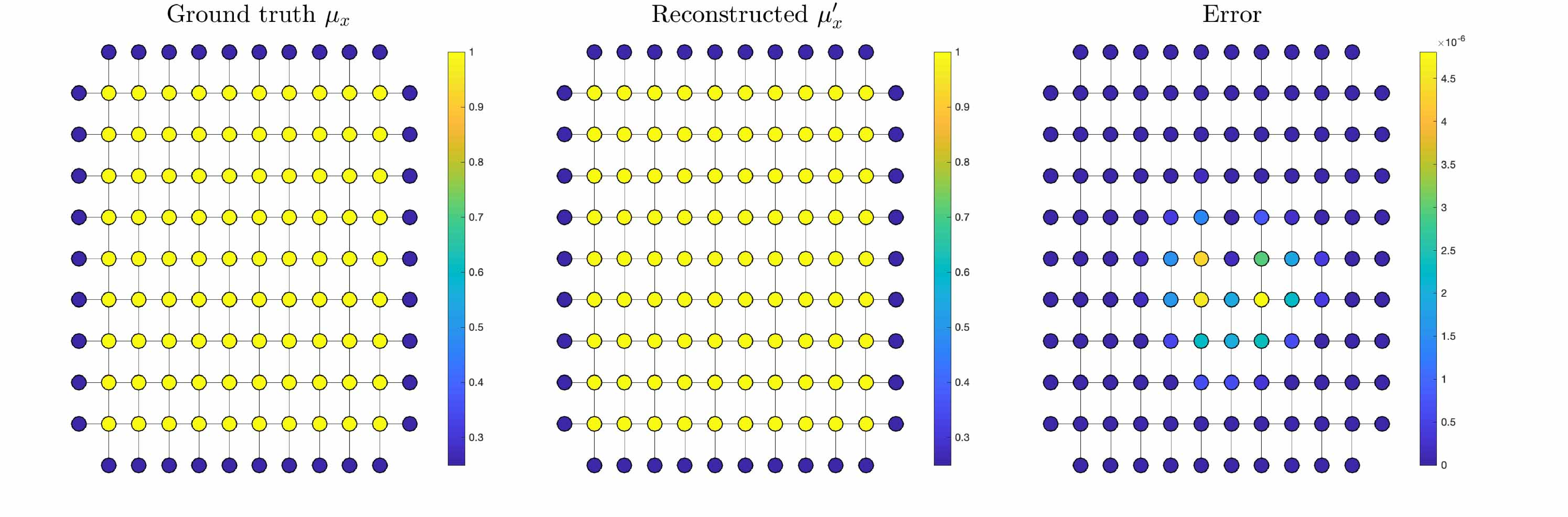}
\caption{\small\small{Experiment 1: The ground-truth $\mu_x$, the reconstructed $\mu_x'$ and the  absolute error.
  $\mathrm {L_2RNE}=1.0983*10^{-4}\%$.}}
\label{fig_case3}
\end{figure}

\bigskip

\textbf{Case 1.2: Gaussian Noise.} Next, we validate the stability of the algorithm by adding Gaussian noise to the Neumann boundary spectral data. The noisy spectral data in use is of the form $\big((1+\varepsilon)\lambda_j,(1+\varepsilon)\phi_{j}|_{\partial G}\big)^{|G|}_{j=1}$, where $\phi$'s are the normalized Neumann eigenfunctions and $\varepsilon \sim \mathcal{N} (0, \sigma)$ is a zero mean Gaussian random variable/vector. We choose $\sigma\in [0.1\%, 0.2\%, 0.5\%]$  in the experiment, respectively. 

In the presence of noise, the FRNEs for reconstructing $\Lambda_\mu$ are $0.086344\%$, $0.17\%$, $3.23\%$, respectively; the FRNEs for reconstructing $[W^{*}W]$ are $8.5562*10^{-2}\%$, $0.16\%$ , $4.78\%$, respectively.  
When applying the truncated SVD to solve \eqref{h0}, the thresholds for singular value truncation are $0.003, 0.005,$  and $0.007$, respectively.
When applying the truncated SVD to solve \eqref{mu}, the thresholds for singular value truncation are $0.001, 0.001,$  and $0.003$, respectively.
Here, different empirical thresholds are taken to achieve optimal results.
The reconstruction $\mu'_x$ and the absolute errors are shown in Fig.\ref{fig_case3_noise}, where the $\mathrm {L_2RNE}$s are $12.6\%$, $12.76\%$, $17.75\%$ respectively.

\begin{figure}[htbp]
\centering
\includegraphics[width=15.5cm,height=11.5cm]{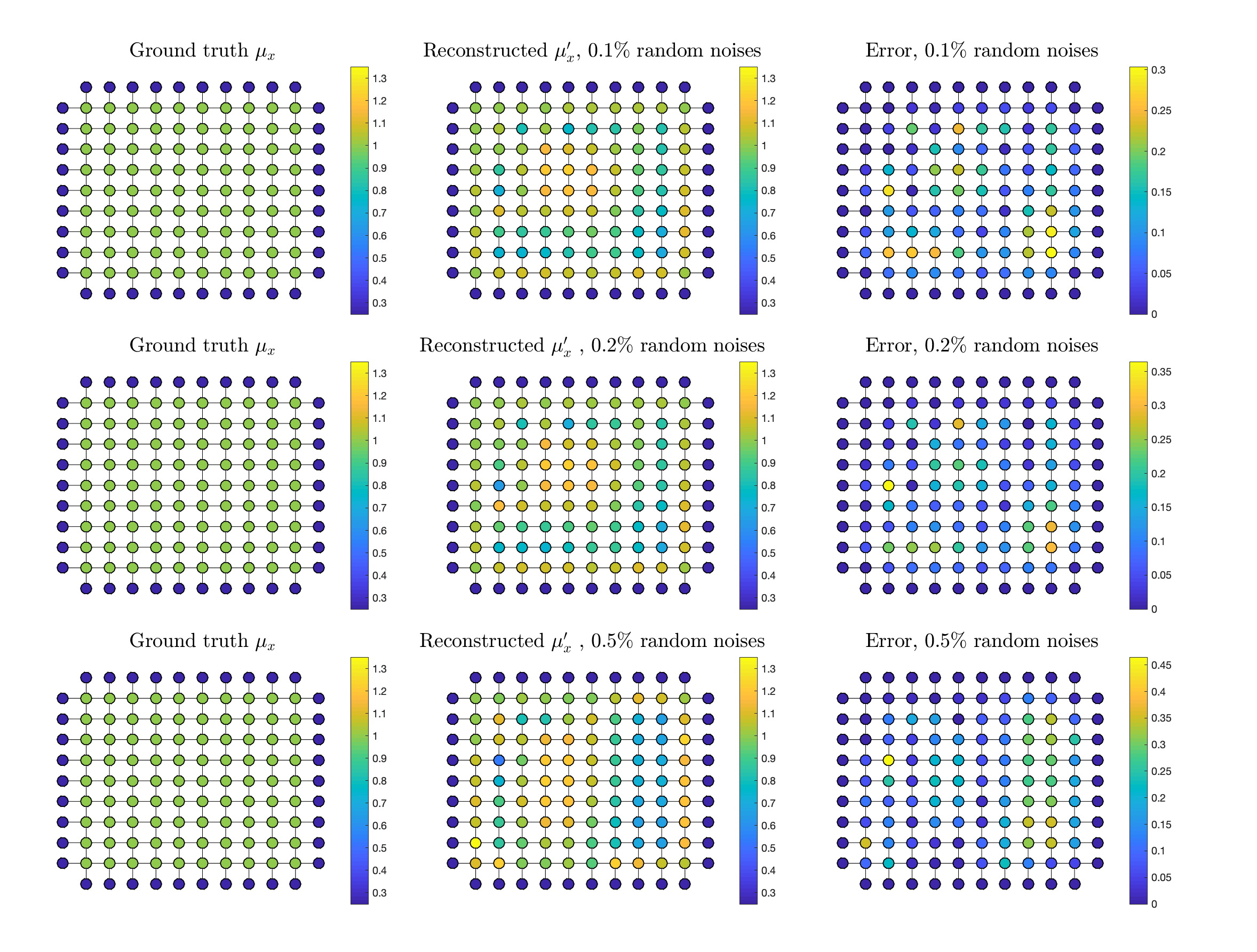}
\caption{\small\small{Experiment 1: Reconstructions and absolute errors in the presence of $0.1\%$, $0.2\%$ and $0.5\%$ Gaussian random noise.
 The $\mathrm {L_2RNE}$s are equal to $12.6\%$, $12.76\%$ and $17.75\%$ respectively.
For comparison, we set the same color bar for the ground truth and the reconstruction.}}
\label{fig_case3_noise}
\end{figure}


\subsection{Experiment 2 : the graph \texorpdfstring{$T_{m,n}$}{}}

In this experiment, we set the following parameters: $m=10, n=9, |\partial G|=38, |G|=90, T=9$, $w_{x,y}=\frac{1}{2}(\mathrm {deg}(x)+\mathrm{deg}(y))$,  and the ground-truth vertex weight is $\mu_x= 1+0.5\sin(x)+0.5\cos(x)$ for all $x = 1,2,\cdots,|\bar{G}|$.

\textbf{Case 2.1: No Noise.} We implement Algorithm~\ref{alg:Framwork} without noise to validate its efficacy.
The $\mathrm {FRNE}$ between the reconstructed ND map $\Lambda_\mu'$ using the Neumann boundary spectral data and the ground truth ND map $\Lambda_\mu$ is  $3.3222*10^{-12}\%$.

When solving the equation (\ref{h0}), the matrix $[W^{*}W]$ is ill conditioned, see
Fig. \ref{fig_three_singular} for its singular values.
We employ the truncated SVD regularization along with the `lsqminnorm' command in MATLAB to find the minimum norm solution as $\vec{h}_0$.
When solving the linear equations (\ref{mu}), we find that ${\rm rank}(\Delta_G)=90$.  
By Lemma~\ref{rank_of_H}, we conclude the vectorized space of harmonic functions $\vec{H}$ has dimension $|\bar{G}|-rank(\Delta_G)=38$.
In this case, from MATLAB, there are
 $128$ linearly independent vectors of the form $\vec{\varphi}\odot\vec{\psi}$ with $\vec{\varphi}, \vec{\psi}\in \vec{H}$.
 We use these $128$ linearly independent vectors as columns to construct the matrix  $[\vec{\varphi}_G \odot \vec{\psi}_G]$ in order to solve (\ref{mu}).
 However, the matrix $[\vec{\varphi}_G \odot \vec{\psi}_G]$ is again ill conditioned, as is shown in  Fig. \ref{fig_three_singular}, so we apply the truncated SVD regularization to find the minimum norm solution.
The reconstruction and the error are shown in Fig. \ref{fig_three}.

\begin{figure}[htbp]
\centering
\includegraphics[scale=0.39]{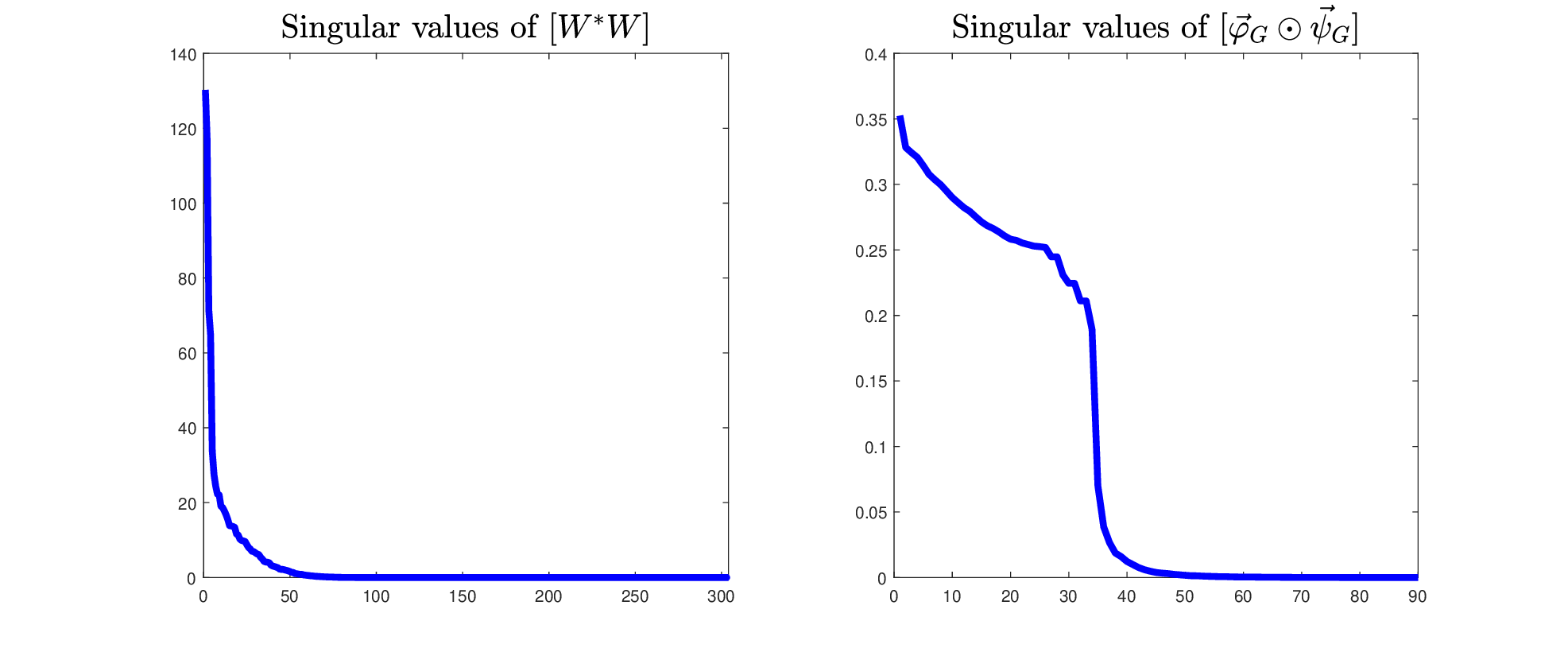}
\caption{\small\small{Experiment 2: The singular values of $[W^{*}W]$ and $[\vec{\varphi}_G \odot \vec{\psi}_G]$. The minimum singular values are $4.2386*10^{-15}$ and $1.1704*10^{-7}$, respectively.}}
\label{fig_three_singular}
\end{figure}

\begin{figure}[htbp]
\centering
\includegraphics[width=15.5cm,height=4.5cm]{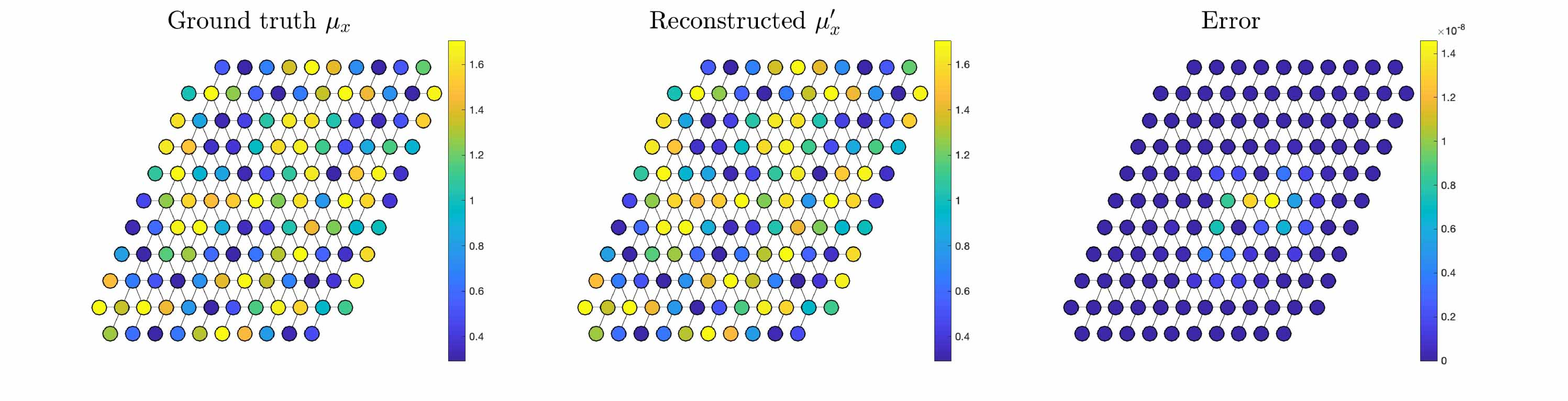}
\caption{\small\small{Experiment 2: The ground-truth $\mu_x$, the reconstructed $\mu_x'$ and the absolute errors.
  $\mathrm {L_2RNE}=2.3817*10^{-7}\%$.}}
\label{fig_three}
\end{figure}

\bigskip

\textbf{Case  2.2: Gaussian Noise.} The noisy spectral data in use is of the form $\big((1+\varepsilon)\lambda_j,(1+\varepsilon)\phi_{j}|_{\partial G}\big)^{|G|}_{j=1}$, where $\phi$'s are the normalized Neumann eigenfunctions and $\varepsilon \sim \mathcal{N} (0, \sigma)$ is a zero mean Gaussian random variable/vector. We choose $\sigma\in [0.1\%, 0.2\%, 0.5\%]$  in the experiment, respectively.

In the presence of noise, the FRNEs for reconstructing $\Lambda_\mu$ are $\mathrm {FRNE}$s are $0.46\%$, $0.94\%$, $2.01\%$, respectively; the $\mathrm {FRNE}$s for reconstructing $[W^{*}W]$ are $0.38\%$, $0.8\%$, $1.02\%$, respectively.
When applying the truncated SVD to solve \eqref{h0}, the thresholds for singular value truncation are $0.001$, $0.005$, and $0.003$, respectively.
When applying the truncated SVD to solve \eqref{mu}, the thresholds for singular value truncation are $0.001, 0.001$, and $0.003$, respectively.
Here, different empirical thresholds are taken to achieve optimal results.
The reconstruction $\mu'_x$ and the absolute errors are shown in Fig. \ref{fig_three_noise}, where the $\mathrm {L_2RNE}$s are $27.81\%$, $28.17\%$ and $32.98\%$, respectively.

 \begin{figure}[htbp]
\centering
\includegraphics[width=15.5cm,height=11.5cm]{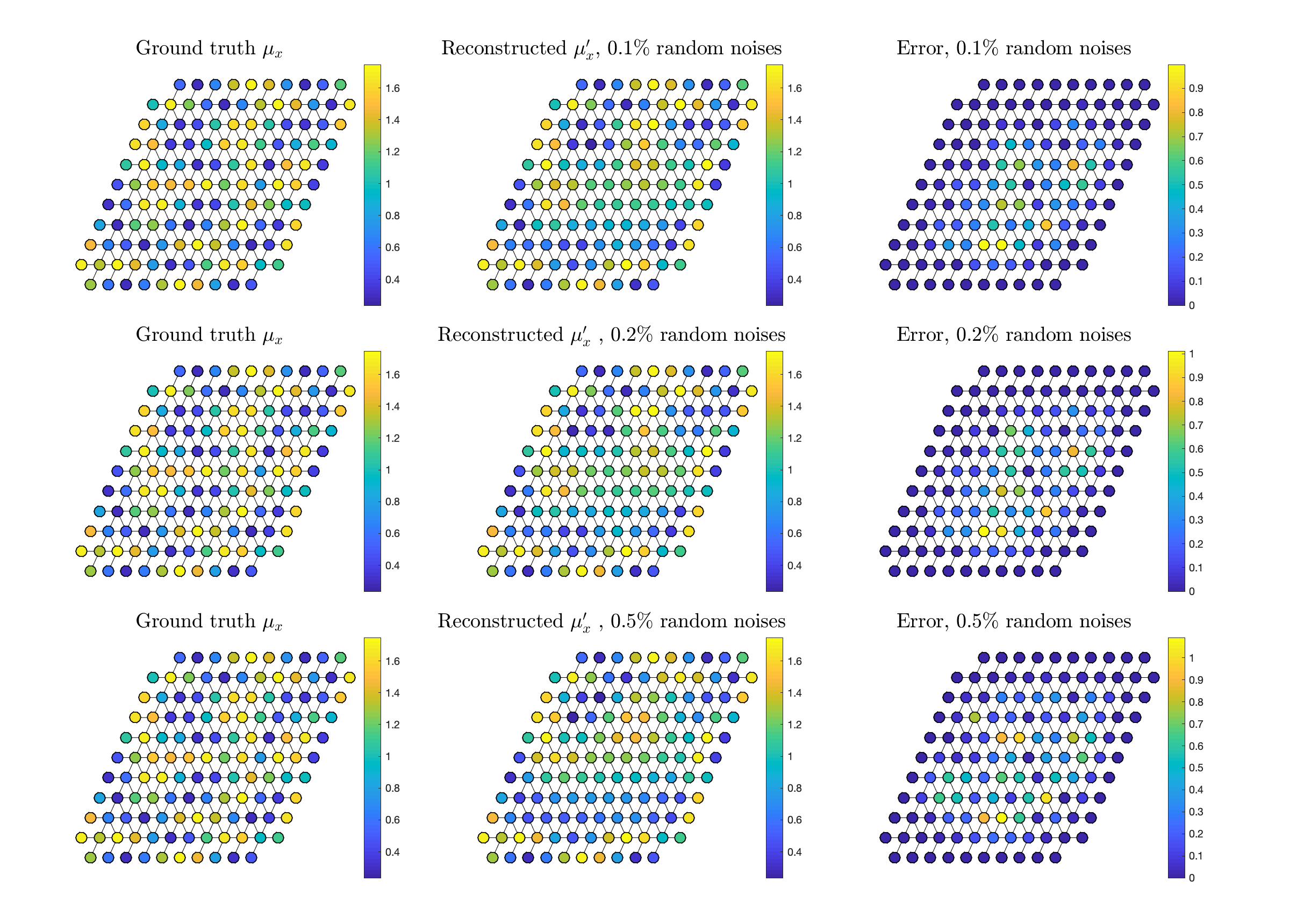}
\caption{\small\small{Experiment 2: Reconstructions and absolute errors in the presence of $0.1\%$, $0.2\%$ and $0.5\%$ Gaussian random noise.  The $\mathrm {L_2RNE}$ are equal to $27.81\%$, $28.17\%$ and $32.98\%$ respectively.
For comparison, we set the same color bar for the ground truth value and the reconstructed value.}}
\label{fig_three_noise}
\end{figure}

\subsection{Experiment 3 : the graph \texorpdfstring{$H_{m,n}$}{}}

In this experiment, we set the following parameteres: $m=9, n=4, |\partial G|=28, |G|=90, T=9$, $w_{x,y}=\frac{1}{2}(\mathrm {deg}(x)+\mathrm{deg}(y))$ and the ground-truth vertex weight is $\mu_x= 1$ for all $x\in \bar{G}$.

\textbf{Case 3.1: No Noise.} We implement Algorithm~\ref{alg:Framwork} without noise to validate its efficacy.
The $\mathrm {FRNE}$ between the reconstructed ND map $\Lambda_\mu'$ using the Neumann boundary spectral data and the ground truth ND map $\Lambda_\mu$ is  $5.3689*10^{-13}\%$.

When solving the equation (\ref{h0}), the matrix $[W^{*}W]$ is ill conditioned, see
Fig. \ref{fig_six_singular} for its singular values.
We employ the truncated SVD regularization along with the `lsqminnorm' command in MATLAB to find the minimum norm solution as $\vec{h}_0$.
When solve the linear equations (\ref{mu}), we find that ${\rm rank}(\Delta_G)=90$.  
By Lemma~\ref{rank_of_H}, we conclude the vectorized space of harmonic functions $\vec{H}$ has dimension $|\bar{G}|-rank(\Delta_G)=28$.
In this case, from MATLAB, there are
 $118$ linearly independent vectors of the form $\vec{\varphi}\odot\vec{\psi}$ with $\vec{\varphi}, \vec{\psi}\in \vec{H}$.
 Here, we use these $118$ linearly independent vectors as columns to construct the matrix  $[\vec{\varphi}_G \odot \vec{\psi}_G]$ in order to solve (\ref{mu}).
 However, the matrix $[\vec{\varphi}_G \odot \vec{\psi}_G]$ is again ill conditioned, as is shown in Fig. \ref{fig_six_singular}, so we apply the truncated SVD regularization to find the minimum norm solution.
The reconstruction and the errors are shown in Fig. \ref{fig_six}.

\begin{figure}[htbp]
\centering
\includegraphics[scale=0.39]{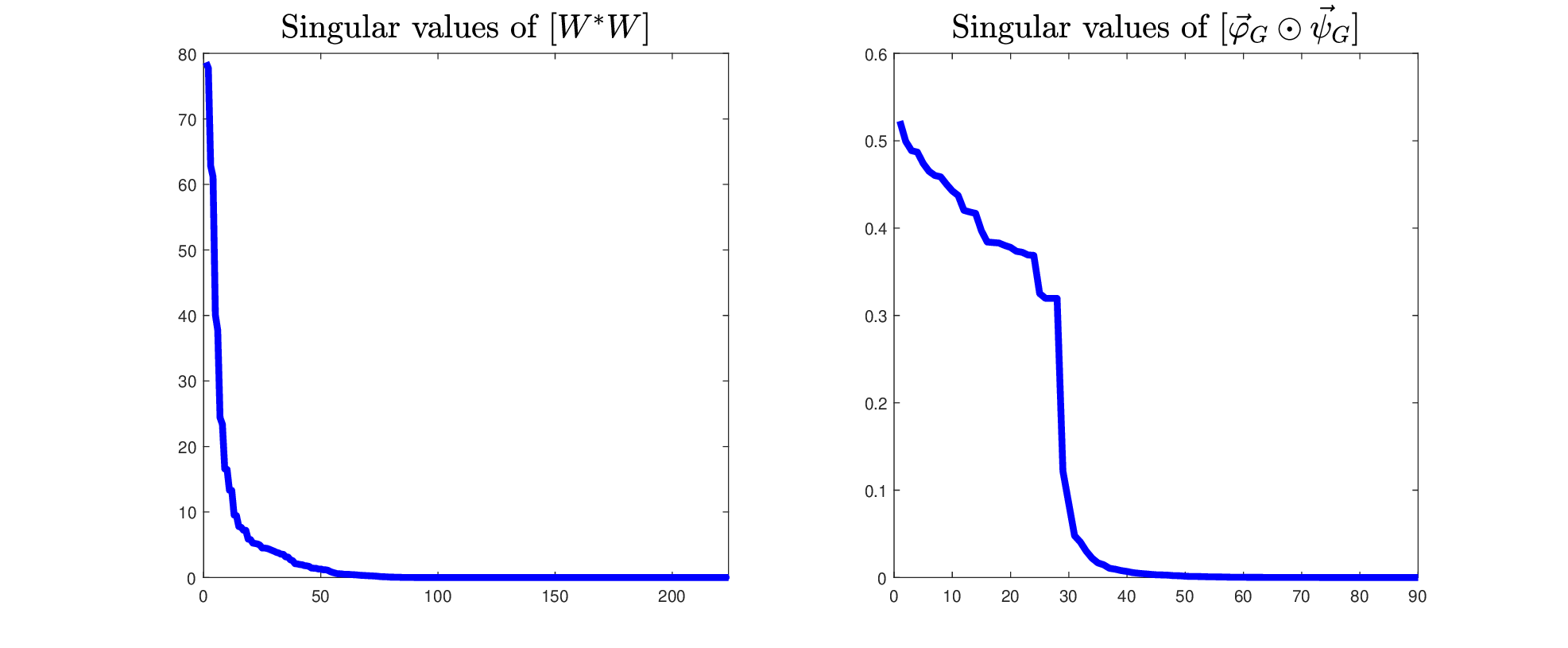}
\caption{\small\small{Experiment 3: The singular values of $[W^{*}W]$ and $[\vec{\varphi}_G \odot \vec{\psi}_G]$. The minimum singular values are $3.7372*10^{-15}$ and $1.1152*10^{-8}$, respectively.}}
\label{fig_six_singular}
\end{figure}

\begin{figure}[htbp]
\centering
\includegraphics[width=15.5cm,height=4.5cm]{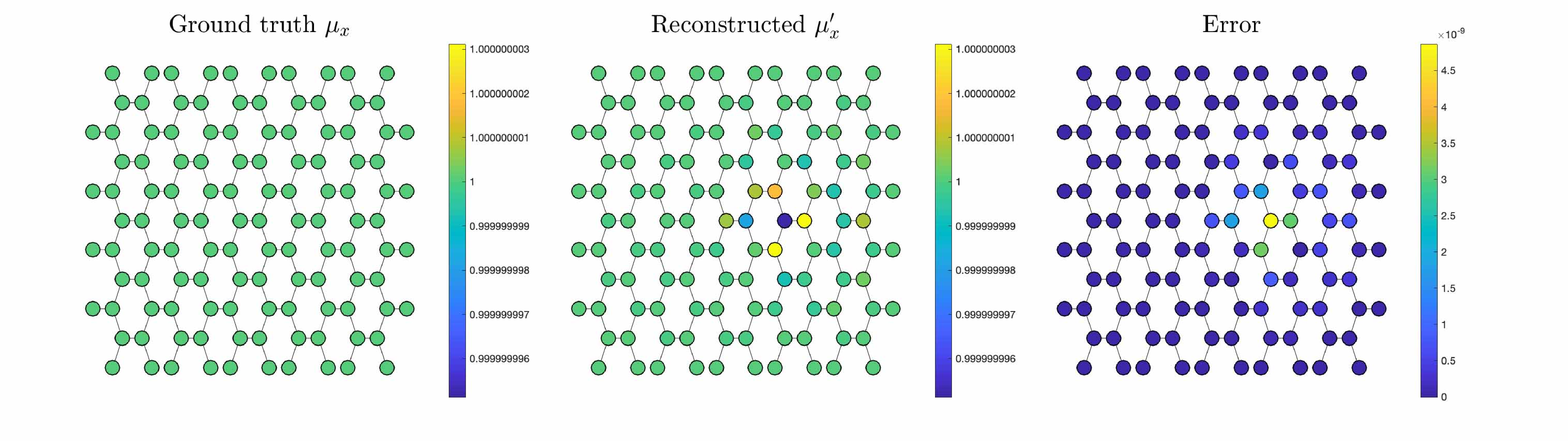}
\caption{\small\small{Experiment 3: The ground-truth $\mu_x$, the reconstructed $\mu_x'$ and the  errors.
  $\mathrm {L_2RNE}=7.7674*10^{-8}\%$.}}
\label{fig_six}
\end{figure}

\textbf{Case  3.2: Gaussian Noise.} The noisy spectral data in use is of the form $\big((1+\varepsilon)\lambda_j,(1+\varepsilon)\phi_{j}|_{\partial G}\big)^{|G|}_{j=1}$, where $\phi$'s are the normalized Neumann eigenfunctions and $\varepsilon \sim \mathcal{N} (0, \sigma)$ is a zero mean Gaussian random variable/vector. We choose $\sigma\in [0.1\%, 0.2\%, 0.5\%]$  in the experiment, respectively.

In the presence of noise, the FRNEs for reconstructing $\Lambda_\mu$ are $0.12\%$, $0.23\%$, $0.67\%$, respectively; the FRNEs for reconstructing $[W^{*}W]$ are $8.8772*10^{-2}\%$, $0.17\%$, $0.47\%$, respectively.  
When applying the truncated SVD to solve \eqref{h0}, the thresholds for singular value truncation are $0.001, 0.005$, and $0.003$, respectively.
When applying the truncated SVD to solve \eqref{mu}, the thresholds for singular value truncation are $0.001, 0.001$, and $0.003$, respectively.
Here, different empirical thresholds are taken to achieve optimal results.
The reconstruction $\mu'_x$ and the absolute errors are shown in Fig. \ref{fig_six_noise}, where the $\mathrm {L_2RNE}$s are $14.59\%$, $16.38\%$ and $24.89\%$ respectively.

 \begin{figure}[htbp]
\centering
\includegraphics[width=15.5cm,height=11.5cm]{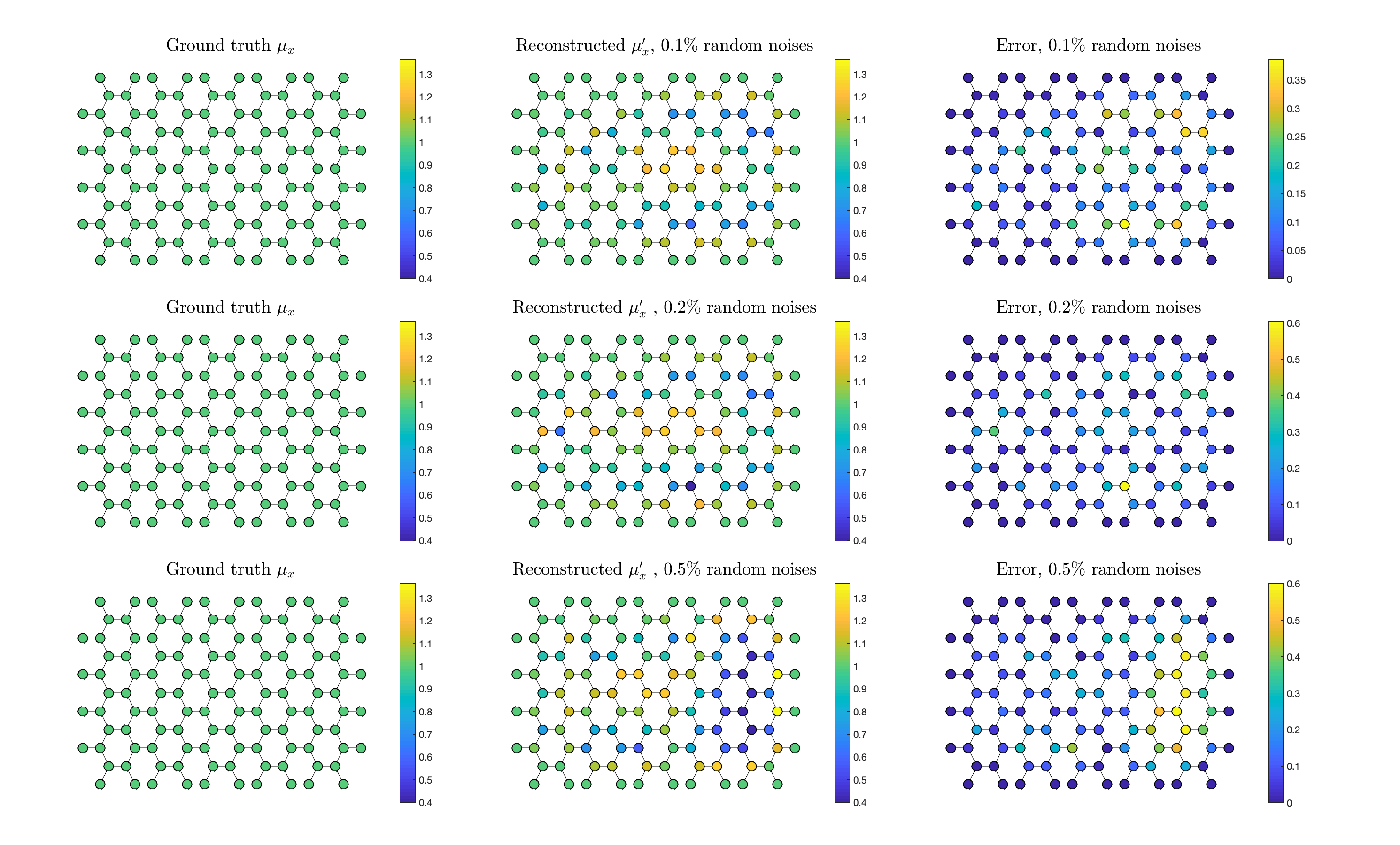}
\caption{\small\small{Experiment 3: Reconstructions and absolute errors in the presence of $0.1\%$, $0.2\%$ and $0.5\%$ Gaussian random noise.  The $\mathrm {L_2RNE}$s are equal to $14.59\%$, $16.38\%$ and $24.89\%$ respectively. For comparison, we set the same color bar for the ground truth value and the reconstructed value.}}
\label{fig_six_noise}
\end{figure}

\begin{appendix}
\section{}\label{appendixA}

In this article, we proved that the Neumann boundary spectral data determines (in a constructive way) the interior vertex weight under Assumption \ref{foliation condition}. On the other hand, the same conclusion is given  in \cite{MR4620352} with different assumptions on graphs.
This appendix compares the two types of assumptions with the goal of highlighting their difference. In particular, we show that neither of the assumptions implies the other. As a result, our assumption identifies a novel class of graphs for which the discrete Gel'fand's inverse spectral problem can be solved.

Recall some definitions and results in~\cite{MR4620352}. 
Let $\mathbb{G}$ be a finite graph with boundary. 
Let $G$ be the set of interior vertices of $\mathbb{G}$.  A subset of these vertices is denoted by $X\subset G$ . 
A vertex $x\in X$ is called an \textit{extreme vertex of $X$ with respect to $\partial G$} if there exists a boundary vertex $z_0\in \partial G$ such that 
\begin{equation*}
d(x,z_0) < \min_{y\in X, \, y\neq x} d(y,z_0).
\end{equation*}
In other words, $x$ is the unique nearest vertex in $X$ to $z_0$.

The major assumption on the graph in~\cite{MR4620352} is the following two conditions:
\begin{enumerate}
	\item Any two interior vertices that are connected to the same boundary vertex are also connected to each other.
	\item (\textit{Two-points condition}) Any subset $X$ with $|X|\geq 2$ has at least two extreme vertices with respect to $\partial G$.
\end{enumerate}
Note that Condition $(1)$ is void when a graph satisfies our Assumption \ref{foliation condition}$(i)$. 
Condition $(2)$ is referred to as the \textit{two-points condition} in~\cite{MR4620352}. Moreover, the following criterion provides sufficient conditions for a graph to satisfy the two-points condition, see \cite[Proposition 1.8]{MR4620352}. 

\begin{lemm}(\cite[Proposition 1.8]{MR4620352}) \label{thm:twoptlemma}
If there exists a function $g: \bar{G}\rightarrow \mathbb{R}$ that satisfies the following conditions:
\begin{enumerate}[(i)]
  \item  $|g(x)-g(y)|\leq 1$ when $x\sim y$;
  \item for every $x\in G$, there is exactly one vertex $y_1\in \mathcal{N}(x)$ such that $g(y_1)-g(x)=1$, and there is exactly one vertex $y_2\in \mathcal{N}(x)$ such that $g(y_2)-g(x)=-1$;
  \item for every $z\in \partial G$, there is at most one vertex $y_3\in \mathcal{N}(z)$ such that $g(y_3)-g(z)=1$, and there is at most one vertex $y_4\in \mathcal{N}(z)$ such that $g(y_4)-g(z)=-1$,
\end{enumerate}
then the graph  is said to satisfy the two-points condition. 
\end{lemm}

We provide two specific graphs to show that the two sets of assumptions are different. First, there exist graphs that satisfy our Assumption~\ref{foliation condition} but not the two-points condition, see Fig. \ref{not belong to they}. This graph satisfies Assumption~\ref{foliation condition} because every vertex has no more than one next-level neighbor. 
However, the subset $X=\{x,y\}$ has just one extreme vertex $x$ with respect to $\partial G$. 
Any path between vertex $y$ and a boundary vertex must contain $x$. Therefore, $y$ cannot be the unique nearest vertex in $X$ to any boundary vertex.

On the other hand, there also exist graphs that satisfy the two-points condition but not our Assumption~\ref{foliation condition}, see Fig. \ref{not belong to we}.
For  ease of notation, we constructed a Cartesian coordinate system in which  the origin is marked, and the vertices are represented by the coordinates $(j,k)\in \mathbb{Z}^2$.
This graph satisfies the two-points condition because the function $g(j,k)=\frac{j}{2}+k$ defined on $\bar{G}$ satisfies all the conditions in Lemma~\ref{thm:twoptlemma}. To demonstrate that it does not satisfy Assumption~\ref{foliation condition}, note that 
\begin{equation*}
N_1 = N_1^1 \cup N_1^2 \cup \{(2,3), (4,2)\}
\end{equation*}
where 
\begin{equation*}
N_1^1=\{(1,0),(1,3),(4,1),(4,5),(2,0),(3,0),(3,5)\}, \qquad N_1^2=\{(1,1)\}.
\end{equation*}
Recall the definition of $N_1^3$ in Assumption \ref{foliation condition}, 
we find that  $(2,3),(4,2)\not\in N_1^3$,
because their next-level neighbors are respectively $(2,2),(3,3)$ and $(3,2),(4,3)$, none of which belong to $\mathcal{N}(N_1^1\cup N_1^2)$. Therefore, the decomposition in Assumption~\ref{foliation condition} does not hold for this graph.

\begin{figure}[htbp]
\center
\subfigure[]{
\begin{minipage}[c]{0.4\linewidth} 
\centering
\includegraphics[width=7.2cm,height=6cm]{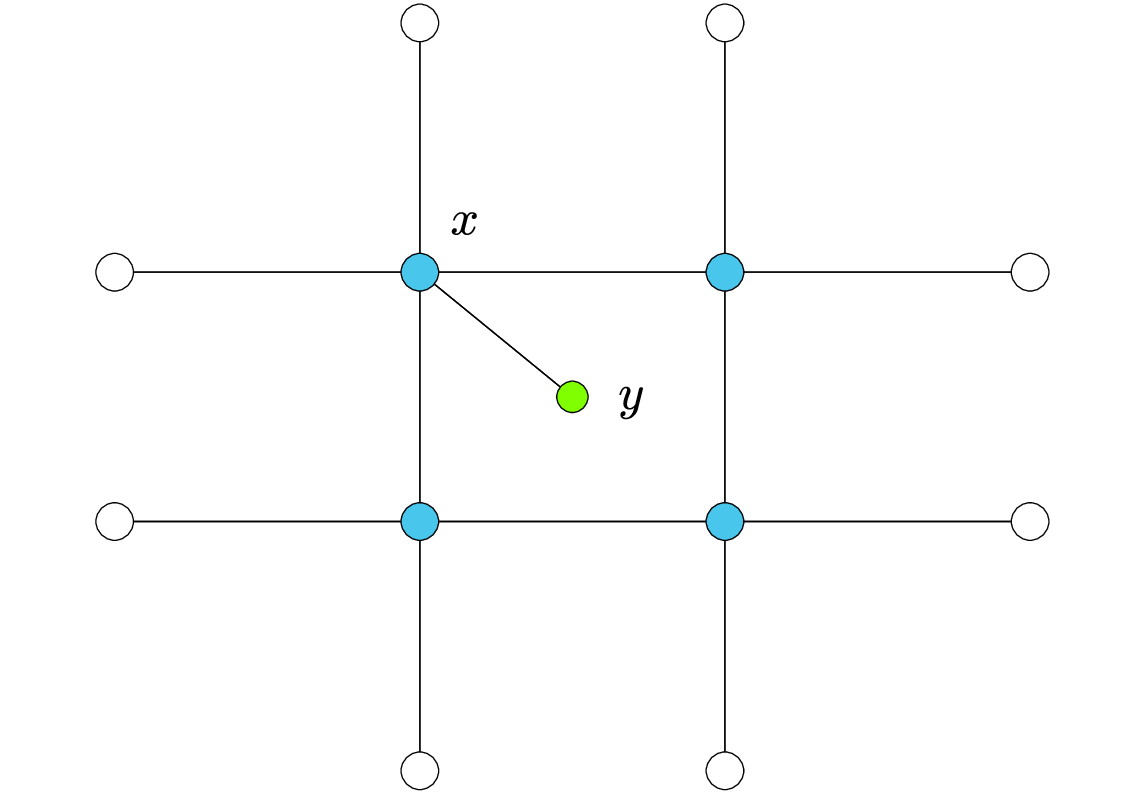} \label{not belong to they}
\hspace{4pt}
\end{minipage}}
\hspace{2pt}
\subfigure[]{
\begin{minipage}[c]{0.4\linewidth} 
\centering
\includegraphics[width=8.6cm,height=6cm]{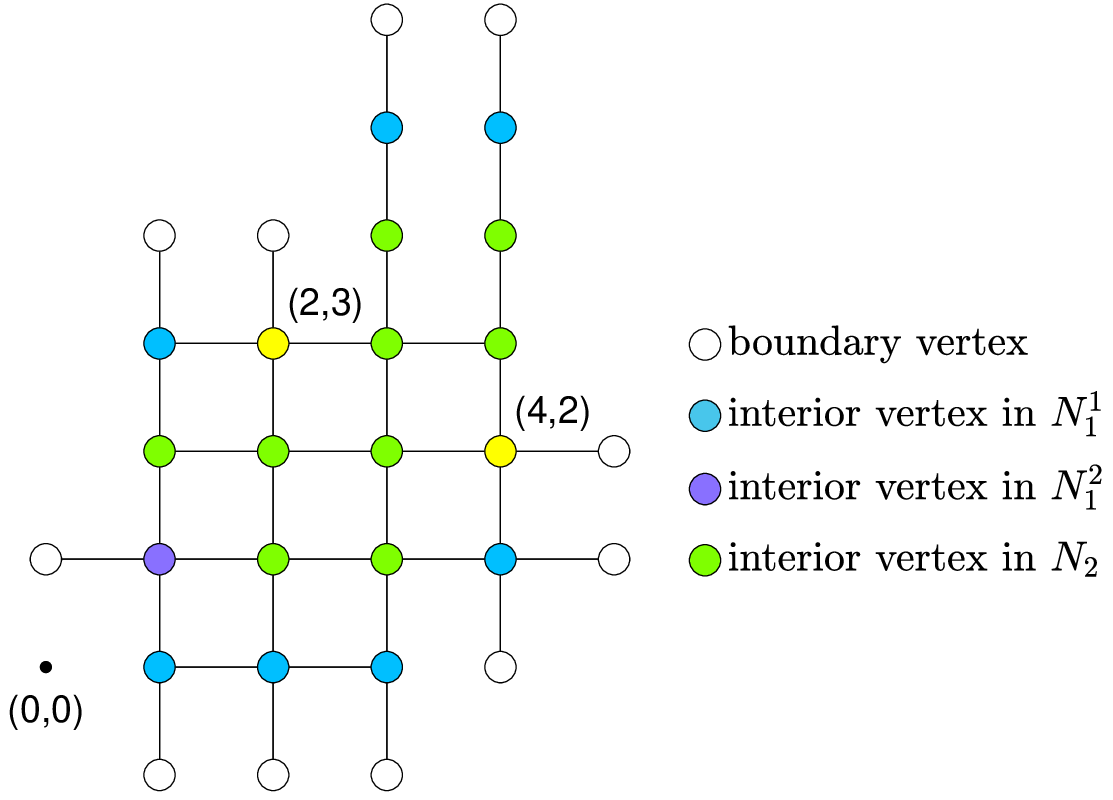}
\label{not belong to we}
\end{minipage}}
\caption{\small\small{(a) A graph satisfies Assumption~\ref{foliation condition} but not the two-points condition.
(b) A graph  satisfies the two-points condition but not Assumption~\ref{foliation condition}.
} }
\end{figure}

\section{}\label{appendixB}

We compute some adjoint operators in this appendix. First, the linear operator  
\begin{equation*}
W: l^2(\{1,\cdots,T-1\}\times\partial G)\longmapsto l^2(G), \qquad 
h\longmapsto u^{P^*_Th}(T,x),~x\in G
\end{equation*}
is introduced in Section~\ref{Prove the reconstruction procedure}.
It maps the Neumann boundary values to the solution of equation \eqref{eq:ibvp1} at time $t=T$ and $x\in G$.

\begin{lemm}
    The adjoint $W^*$ is given by
\[ W^*g = v(t,z), \qquad  (t,z)\in\{{1,2,\cdots,T-1\}\times\partial G} \] 
where $v$ satisfies the following problem:
\begin{align*}
\begin{cases}
D_{tt} v(t,x) - \Delta_Gv(t,x) = 0, & (t,x)\in\{1,2,\cdots, T-1\}\times G, \\
v(T,x) = 0, & x\in \bar{G}, \\
D_t v(T-1,x) =  g(x), & x\in G, \\
\partial_\nu v(t,z)  = 0, & (t,z)\in \{0,1,\cdots, T\}\times \partial G.
\end{cases}
\end{align*}
\end{lemm}

\begin{proof}
Let $u$ be the solution of \eqref{eq:ibvp1}.
As $v$ satisfies the graph wave equation above, we have 
\begin{align*}
0&= \sum\limits_{x\in G}\mu_x\sum\limits_{t=1}^{T-1}  (D_{tt} v(t,x) - \Delta_G v(t,x))u(t,x)\\
&= \sum\limits_{x\in G}\mu_x\sum\limits_{t=1}^{T-1}D_{tt} v(t,x)u(t,x)-\sum\limits_{x\in G}\mu_x\sum\limits_{t=1}^{T-1}u(t,x)\Delta_G v(t,x)\\
&:= I_1-I_2.
\end{align*}
For $I_1$, using  the definition of the operators $D_t$ and $D_{tt}$, we can obtain
\begin{align*}
I_1
=&\sum\limits_{x\in G}\mu_x\sum\limits_{t=1}^{T-1}v(t+1,x)u(t,x)-2\sum\limits_{x\in G}\mu_x\sum\limits_{t=1}^{T-1}v(t,x)u(t,x)+\sum\limits_{x\in G}\mu_x\sum\limits_{t=1}^{T-1}v(t-1,x)u(t,x)\\
=&\sum\limits_{x\in G}\mu_x\sum\limits_{t=1}^{T-1}(v(t,x)u(t-1,x)-2v(t,x)u(t,x)+v(t,x)u(t+1,x))\\
&+\sum\limits_{x\in G}\mu_x(-v(1,x)u(0,x)+v(T,x)u(T-1,x)+v(0,x)u(1,x)-v(T-1,x)u(T,x))\\
=&\sum\limits_{x\in G}\mu_x\sum\limits_{t=1}^{T-1}v(t,x)D_{tt}u(t,x)+\sum\limits_{x\in G}\mu_xg(x)u(T,x),
\end{align*}
where  we have used the fact that 
$u(0,x)=u(1,x)=0$, $v(T,x)=0$ and $v(T-1,x)=-g(x)$ for $x\in G$.

For $I_2$, we compute it using Lemma~\ref{Green's formula on graph}:
\begin{align*}
I_2 &= \sum\limits_{x\in G}\mu_x\sum\limits_{t=1}^{T-1}v(t,x)\Delta_G u(t,x)+ \sum\limits_{z\in \partial G}\mu_z\sum\limits_{t=1}^{T-1}(v(t,z)\partial_\nu u(t,z)-u(t,z)\partial_\nu v(t,z))\\
&= \sum\limits_{x\in G}\mu_x\sum\limits_{t=1}^{T-1}v(t,x)\Delta_G u(t,x)+ \sum\limits_{z\in \partial G}\mu_z\sum\limits_{t=1}^{T-1}v(t,z)f(t,z),
\end{align*}
where  we have used the boundary conditions $\partial_\nu v(t,z)=0$ and $\partial_\nu u(t,z)=f$ for $z\in\partial G$.

As $I_1-I_2=0$ and the first terms of $I_1$ and $I_2$ are identical, we conclude:
\begin{align*}
\sum_{z\in \partial G}\mu_z\sum_{t=1}^{T-1}v(t,z)f(t,z) = \sum\limits_{x\in G}\mu_xg(x)u(T,x) = (u(T),g)_G=(Wf,g)_G.
\end{align*}
The proof is complete when we observe that the left hand side is exactly $(f,v)_{\{1,\cdots,T-1\}\times\partial G}$.

\end{proof}

\subsection{Calculate the map  \texorpdfstring{$\Lambda_{\mu,T}^*$}{}} 

Next, we derive the adjoint to the operator $\Lambda_{\mu,T}$ introduced in Section~\ref{Prove the reconstruction procedure}. Recall that
\begin{align*}
\Lambda_{\mu,T} h:= P_T(\Lambda_\mu P_T^*h)=u^{P_T^*h}|_{\{1,\cdots,T-1\}\times\partial G}
\end{align*}
for $h\in l^2({\{1,\cdots,T-1\}\times\partial G})$.

\begin{lemm}
    The adjoint $\Lambda^*_{\mu,T}$ is given by
 \begin{equation*}   \Lambda_{\mu,T}^*=\mathscr{R}\Lambda_{\mu,T}\mathscr{R}.
\end{equation*}
\end{lemm}

\begin{proof}
Let $u$ be the solution of~\eqref{eq:ibvp1} and $v$ be the solution of 
\begin{align*}
\begin{cases}
D_{tt} v(t,x) - \Delta_G v(t,x) = 0, & (t,x)\in \{1,\cdots, T-1\}\times G, \\
v(T-1,x) = 0, & x\in \bar{G}, \\
D_t v(T-1,x) = 0, & x\in G,\\
\partial_\nu v(t,z)  = g(t,z), & (t,z)\in \{0,\cdots, T\}\times \partial G.
\end{cases}
\end{align*}
Using Lemma \ref{Green's formula on graph}, we obtain
the equation
\begin{align}\label{rl}
0= &\sum\limits_{x\in G}\mu_x\sum\limits_{t=1}^{T-1}  (D_{tt} v(t,x) - \Delta_G v(t,x))u(t,x)\nonumber\\
=&\sum\limits_{x\in G}\mu_x\sum\limits_{t=1}^{T-1}D_{tt} v(t,x)u(t,x)-\sum\limits_{x\in G}\mu_x\sum\limits_{t=1}^{T-1}\Delta_G v(t,x)u(t,x)\nonumber\\
=&\sum\limits_{x\in G}\mu_x\sum\limits_{t=1}^{T-1}(v(t,x)u(t-1,x)-2v(t,x)u(t,x)+v(t,x)u(t+1,x))\nonumber\\
&+\sum\limits_{x\in G}\mu_x(-v(1,x)u(0,x)+v(T,x)u(T-1,x)+v(0,x)u(1,x)-v(T-1,x)u(T,x))\nonumber\\
&-\sum\limits_{x\in G}\mu_x\sum\limits_{t=1}^{T-1}v\Delta_G u+\sum\limits_{z\in \partial G}\mu_z\sum\limits_{t=1}^{T-1}(\partial_\nu v(t,z)u(t,z)-v(t,z)\partial_\nu u(t,z))\nonumber\\
=& \sum\limits_{x\in G}\mu_x\sum\limits_{t=1}^{T-1}v(t,x)(D_{tt}u(t,x)-\Delta_G u(t,x))+\sum\limits_{z\in \partial G}\mu_z\sum\limits_{t=1}^{T-1}(\partial_\nu v(t,z)u(t,z)-v(t,z)\partial_\nu u(t,z))\nonumber\\
=&\sum\limits_{z\in \partial G}\mu_z\sum\limits_{t=1}^{T-1}g(\Lambda_{\mu,T}f)- \sum\limits_{z\in \partial G}\mu_z\sum\limits_{t=1}^{T-1}v(t,z)f.
\end{align}

On the other hand, consider $U$ that  satisfies  the following problem
\begin{align*}
\begin{cases}
D_{tt} U(t,x) - \Delta_G U(t,x) = 0,  & (t,x)\in \{1,\cdots, T-1\}\times G, \\
U(0,x) = 0, & x\in \bar{G}, \\
D_t U(0,x) = 0, & x\in G,\\
\partial_\nu U(t,z)  = \mathscr{R} g(t,z), & (t,z)\in \{0,\cdots, T\}\times \partial G.
\end{cases}
\end{align*}
Then, $\mathscr{R} U = v$, since they solve the initial boundary value problem. We conclude
\begin{align*}
\mathscr{R}\Lambda_{\mu,T}\mathscr{R}g=\mathscr{R}(U|_{\{1,\cdots,T-1\}\times\partial G})=v|_{\{1,\cdots,T-1\}\times\partial G}.
\end{align*}
Substitute this relation into  equation \eqref{rl} to get 
\begin{align*}
0&=\sum\limits_{z\in \partial G}\mu_z\sum\limits_{t=1}^{T-1}g(\Lambda_{\mu,T}f)- \sum\limits_{z\in \partial G}\mu_z\sum\limits_{t=1}^{T-1}f\mathscr{R}\Lambda_{\mu,T}\mathscr{R}g\\
&=(\Lambda_{\mu,T}f,g)_{\{1,\cdots,T-1\}\times\partial G}-(f,\mathscr{R}\Lambda_{\mu,T}\mathscr{R}g)_{\{1,\cdots,T-1\}\times\partial G}
\end{align*}
for all $f,g$. This completes the proof.

\end{proof}

\section{}\label{appendixC}
\subsection{Matrix form of the graph Laplacian operator \texorpdfstring{$\Delta_G$}{}}

In this appendix, we use the matrix form of the graph Laplacian to prove a few auxiliary results. As usual, we index the interior vertices by $x_1,x_2,\cdots,x_{|G|}$ and the boundary vertices by $x_{|G|+1},x_{|G|+2},\cdots,x_{|\bar{G}|}$ on $\bar{G}$. 
Recall that in Section~\ref{sec:uniqueandrecon}, the graph Laplacian $\Delta_G: \bar{G}\rightarrow G$ is identified with a block matrix
\begin{equation*}
[\Delta_G] = ([\Delta_{G,G}], [\Delta_{G,\partial G}]) \;\in \mathbb{R}^{|G|\times |\bar{G}|}, \quad \text{ where }
[\Delta_{G,G}] \in \mathbb{R}^{|G|\times |G|}, \; 
[\Delta_{G,\partial G}] \in \mathbb{R}^{|G|\times |\partial G|}.
\end{equation*}
Then, the graph Laplacian operator $\Delta_G$ can be written as matrix form as follows
\begin{align*}
\begin{split} 
\begin{pmatrix}
\frac {-\sum\limits_{\substack{y\in \bar{G}\\ y\sim {x_1}}} w(x_1,y)} {\mu_{x_1}}  & \frac {w(x_1,x_2)} {\mu_{x_1}} & \cdots & \frac {w(x_1,x_{|G|})}{\mu_{x_1}} & \cdots & \frac {w(x_1,x_{|\bar{G}|})} {\mu_{x_1}}\\
\frac {w(x_2,x_1)} {\mu_{x_2}} & \frac {-\sum\limits_{\substack{y\in \bar{G}\\ y\sim {x_2}}} w(x_2,y)} {\mu_{x_2}}  & \cdots & \frac {w(x_2,x_{|G|})} {\mu_{x_2}} & \cdots & \frac {w(x_2,x_{|\bar{G}|})} {\mu_{x_2}}\\
\vdots & \vdots & & \vdots & & \vdots\\
\frac {w(x_{|G|},x_1)} {\mu_{x_{|G|}}} & \frac {w(x_{|G|},x_2)} {\mu_{x_{|G|}}}  & \cdots & \frac {-\sum\limits_{\substack{y\in \bar{G}\\ y\sim {x_{|G|}}}} w(x_{|G|},y)} {\mu_{x_{|G|}}} & \cdots & \frac {w(x_{|G|},x_{|\bar{G}|})} {\mu_{x_{|G|}}}\\
\end{pmatrix}
:=
\begin{pmatrix}
\Delta_{G,G}  & \vdots & \Delta_{G,\partial G}
\end{pmatrix},
\end{split}
\end{align*}
where
\begin{align*}
\begin{split}
\Delta_{G,G}=
\begin{pmatrix}
\frac{1}{\mu_{x_1}} &  \\
 & \frac{1}{\mu_{x_2}}\\
 & & \ddots &\\
 & & & \frac{1}{\mu_{x_{|G|}}}
\end{pmatrix}
\begin{pmatrix}
-\sum\limits_{\substack{y\in \bar{G}\\ y\sim {x_1}}} w(x_1,y)  & w(x_1,x_2) & \cdots & w(x_1,x_{|G|})\\
w(x_2,x_1) & -\sum\limits_{\substack{y\in \bar{G}\\ y\sim {x_2}}} w(x_2,y)  & \cdots & w(x_2,x_{|G|})\\
\vdots & \vdots & & \vdots\\
w(x_{|G|},x_1) & w(x_{|G|},x_2) & \cdots & -\sum\limits_{\substack{y\in \bar{G}\\ y\sim {x_{|G|}}}} w(x_{|G|},y)
\end{pmatrix},
\end{split}
\end{align*}
and
\begin{align*}
\begin{split}
\Delta_{G,\partial G}=
\begin{pmatrix}
\frac{1}{\mu_{x_1}} &  \\
 & \frac{1}{\mu_{x_2}}\\
 & & \ddots &\\
 & & & \frac{1}{\mu_{x_{|G|}}}
\end{pmatrix}
\begin{pmatrix}
w(x_1,x_{|G|+1}) & \cdots &  w(x_1,x_{|\bar{G}|})\\
w(x_2,x_{|G|+1}) & \cdots & w(x_2,x_{|\bar{G}|})\\
\vdots & & \vdots\\
w(x_{|G|},x_{|G|+1}) & \cdots & w(x_{|G|},x_{|\bar{G}|})
\end{pmatrix}.
\end{split}
\end{align*}
Since the edge weight function $w(\cdot,\cdot)$ is symmetric, the resulting matrix $\Delta_{G,G}$ is also symmetric.

\begin{lemm} \label{thm:ellipticBVP}
For any function $g:\partial G\rightarrow\mathbb{R}$, the boundary value problem
\begin{equation*}
    \Delta_G \varphi(x) = 0~\text{ for }~x\in G, \qquad \varphi|_{\partial G}=g
\end{equation*}
has a unique solution $\varphi:\bar{G}\to \mathbb{R}$.
\end{lemm}

\begin{proof}
Using the vectorization $\vec{\varphi} = (\vec{\varphi}|_G, \vec{g})^T$ and the matrix $[\Delta_G] = ([\Delta_{G,G}], [\Delta_{G,\partial G}])$, the boundary value problem is equivalent to the homogeneous linear system
\begin{equation*}
\left(
\begin{array}{cc}
   [\Delta_{G, G}],  &  [\Delta_{G,\partial G}] 
\end{array}
\right)
\left(
\begin{array}{c}
    \vec{\varphi}|_{G}  \\
    \vec{g}
\end{array}
\right) = 0.
\end{equation*}
Since the matrix $[\Delta_{G,G}]$ is non-singular~\cite[Lemma 3.8]{Curtis2000InversePF}, the linear system admits a unique solution
$$
\vec{\varphi}|_G = -[\Delta_{G,G}]^{-1} [\Delta_{G,\partial G}] \vec{g}.
$$
\end{proof}

\begin{lemm} \label{app:realanalytic}
Let $\beta$ be an arbitrary selection of $|G|$ columns from $\mathbf{H}$. If $\det(\mathbf{H}_{:,\beta})\neq 0$, then $\det(\mathbf{H}_{:,\beta})$ is a real analytic function of $\{w_{x,y}\}\in \mathbb{R}_+^{|\mathcal{E}|}$. 
\end{lemm}

\begin{proof}
Using the cofactor formula, each entry of $[\Delta_{G,G}]^{-1}$ is a rational function of 
$\{w_{x,y}\}\in\mathbb{R}_+^{|\mathcal{E}|}$ (since $\det[\Delta_{G,G}]$ is a polynomial of the entries). 
Recall that for polynomials $A_1(a)$ and $A_2(a)\neq0$ when $a\in \mathbb{R}^{|\mathcal{E}|}$, rational functions of the form $\frac{A_1(a)}{A_2(a)}$ are analytic on any connected subset of $\mathbb{R}^{|\mathcal{E}|}$.
Since  $\Delta_{G,G}$ is invertible, the  functions $\det(\mathbf{H}_{:,\beta})$ are real analytic with respect to $\{w_{x,y}\}\in \mathbb{R}_+^{|\mathcal{E}|}$.
\end{proof}

\end{appendix}

\bibliographystyle{abbrv}
\bibliography{GYL}

\end{document}